\documentclass[a4paper]{article}
\usepackage{amsmath, amsthm, amssymb}
\usepackage{indentfirst}
\usepackage[latin1]{inputenc}
\usepackage[normalem]{ulem}
\usepackage{cancel}
\usepackage[bookmarks=false,colorlinks=true,linkcolor=black,citecolor=black,filecolor=black,urlcolor=black]{hyperref}

\newtheorem{theo}{Theorem}[section]
\newtheorem{lemm}[theo]{Lemma}
\newtheorem{coro}[theo]{Corollary}
\newtheorem{prop}[theo]{Proposition}
\newtheorem{defi}[theo]{Definition}
\newtheorem{rema}[theo]{Remark}
\newtheorem{example}[theo]{Example}

\newcommand{\w}{\mbox{wt}}
\newcommand{\ba}{{\mathbf a}}
\newcommand{\bb}{{\mathbf b}}
\newcommand{\bu}{\mathbf{u}}

\newcommand{\be}{\mathbf{e}}
\newcommand{\zero}{{\mathbf 0}}
\newcommand{\one}{{\mathbf 1}}
\newcommand{\Z}{\mathbb{Z}}
\newcommand{\cA}{{\mathcal{A}}}
\newcommand{\cC}{{\mathcal{C}}}
\newcommand{\cD}{{\mathcal{D}}}
\newcommand{\cQ}{{\mathcal{Q}}}

\newcommand{\cG}{{\mathcal{G}}}

\newcommand{\inv}{^{-1}}
\newcommand{\gen}[1]{\langle #1 \rangle}
\newcommand{\genset}{$x_1,\dots,x_{\sigma}$; $y_1,\dots,y_{\delta}$; $z_1,\dots,z_{\rho}$~}
\newcommand{\qh}{qua\-ter\-nio\-nic Ha\-da\-mard~}
\newcommand{\zz}{$\Z_2\Z_4$-linear }
\newcommand{\qq}{\Z_2\Z_4\cQ_8}
\newcommand{\matriz}[1]{\begin{array} #1 \end{array}}
\newcommand{\Aut}{\mbox{\rm Aut}}

\title{Families of Hadamard $\qq$-codes\thanks{This work has been partially supported by the Spanish MICINN under Grants MTM2009-08435
and by the Catalan AGAUR under Grant 2009SGR1224. \newline
\indent $^1$\'{A}ngel del Rio is with the Department of Mathematics, Universidad de Murcia, Spain. (email: adelrio@um.es)\newline
\indent $^2$J. Rif\`{a} is with the Department of Information and Communications Engineering, Universitat Aut\`{o}noma de Barcelona, Spain. (email:~josep.rifa@uab.cat)}}

\author{\'{A}. del Rio$^1$, J. Rif\`{a}$^2$}

\begin{document}
\maketitle
\begin{abstract}
A $\qq$-code $\cC$ is a non-empty subgroup of $\Z_2^{k_1}\times \Z_4^{k_2}\times \cQ_8^{k_3}$, where $\cQ_8$ is the quaternion group on eight elements. Such $\qq$-codes are translation
invariant propelinear codes as the well known $\Z_4$-linear or \zz codes.

In the current paper, we show that there exist
``pure'' $\qq$-codes, that is, codes that
do not admit any abelian translation invariant propelinear
structure. We study the dimension of the kernel and rank of the $\qq$-codes, and we give upper and lower bounds for these parameters. We give tools to construct a new class of Hadamard codes formed by several families of $\qq$-codes; we study and show the different shapes of such a codes and we improve the upper and lower bounds for the rank and the dimension of the kernel when the codes are Hadamard.
\end{abstract}

\section{Introduction}
The discovery of the existence of a quaternary structure in some relevant families with better parameters than any linear code has raised the interest in the study of these codes \cite{z4} and more generally on codes with a group structure. From the Coding Theory perspective it is desired that the group operation preserves the Hamming distance. This is the case, for example, of the \zz codes which has been intensively studied during last years. More generally, the propelinear codes and, specially those which are translation invariant, are particularly interesting because both left and right product preserves the Hamming distance. Translation invariant propelinear codes has been characterized as the image of a subgroup by a suitable Gray map of a direct product of $\Z_2$, $\Z_4$ and $\cQ_8$, the quaternion group of order 8 \cite{rp}. Hence it makes sense to call this codes as $\Z_2\Z_4\cQ_8$-codes. The aim of this paper is to study the structure and main properties of $\Z_2\Z_4\cQ_8$-codes with special focus on those that are Hadamard codes as well.

Section~\ref{sec:preliminaris} has been reserved for notation and preliminaries.

As far a we know there is not any example in the literature of a proper $\qq$-code, i.e., one which is not equivalent to a \zz code. The first result of this paper consists in providing such an example. This result appears in Section~\ref{sec:Z2Z4Q8}, where we also study the group-theoretical properties of the $\Z_2\Z_4\cQ_8$-codes and the relation of this structure with its rank and the dimension of its kernel.
This structure suggests to associate to the group three numerical parameters. We will show that these parameters provide bounds for the rank and dimension of the kernel. Moreover, our example of proper $\Z_2\Z_4\cQ_8$-code shows that these bounds are tight.

Section~\ref{sec:Hadamard} is dedicated to Hadamard \zz codes.
The Hadamard \zz codes as well as the (extended) perfect \zz codes are well known~\cite{br,kr,PRV06}.
The Hadamard linear codes are dual of extended perfect codes.
However we will show that the extended perfect codes, involving at least one quaternionic component do not exists for length $n\geq 8$. For every $n=2^m$ there is a unique Hadamard linear code, up to equivalence.
If $m\le 3$ this is the unique Hadamard code.
However, there are five inequivalent Hadamard codes of length 16. One of them is cyclic, another is a \zz code and the other three cannot be realized as \zz codes. We will show that exactly one of these three can be realized as a $\Z_2\Z_4\cQ_8$-code, more specifically, as a pure $\cQ_8$-code. This provides another example of such codes.
In Theorem~\ref{ClasificacionHadamard} we provide a precise description of the possible group structures of a Hadamard $\Z_2\Z_4\cQ_8$-codes and in Corollary~\ref{Cotas} we obtain bounds for the rank of a Hadamard $\Z_2\Z_4\cQ_8$-code which are better than those for general $\Z_2\Z_4\cQ_8$-codes.

In the last section of the paper we introduce two constructions of $\Z_2\Z_4\cQ_8$-codes which allow to construct many Hadamard $\Z_2\Z_4\cQ_8$-code

\section{Preliminaries}\label{sec:preliminaris}

Let $\Z_2$ and $\Z_4$ denote the binary field and the ring of integers modulo 4, respectively. Let $\Z_2^n$ denote the set of all binary
vectors of length $n$ and let $\Z_4^n$ be the set of all $n$-tuples
over the ring $\Z_4$. The all-zero vector in $\Z^n_2$ is denoted by
${\mathbf 0}$. Let $\w(v)$ denote the {\em Hamming weight} of a
vector $v \in \Z_2^n$ (i.e., the number of its nonzero coordinates),
and let $d(v, u)=\w(v+u)$, the {\em Hamming distance} between two
vectors $v,u\in\Z_2^n$.

Any non-empty subset of $\Z_2^n$ is called a binary code and a linear subspace of
$\Z_2^n$ is called a {\it binary linear code} or a {\it
$\Z_2$-linear code}. Similarly, any non-empty subset of $\Z_4^n$
is a quaternary code and a subgroup of $\Z_4^n$ is called a {\it
quaternary linear code} \cite{z4}.
Quaternary codes can be viewed as binary codes under the Gray map
defined as
$$
\varphi(0)=(0,0),\, \varphi(1)=(0,1), \ \varphi(2)=(1,1),\,
\varphi(3)=(1,0),
$$
which is extended coordinatewise to a bijection $\phi:\Z_4^n\rightarrow \Z_2^{2n}$.
If $\cC$ is a quaternary linear code of length $n$, then the binary code
$C=\varphi(\cC)$ is said to be a {\em
$\Z_4$-linear code} of binary length $2n$ \cite{z4}.

Let $\cQ_8$ be the {\em quaternion group} on eight elements. The following equalities provides a presentation and the list of elements of $\cQ_8$:
$$\cQ_8=\gen{\ba,\bb|\ba^4=\ba^2\bb^2=\one,\bb\ba\bb^{-1}=\ba^{-1}} =\{\one,\ba,\ba^2,\ba^3,\bb,\ba\bb,\ba^2\bb,\ba^3\bb\}.$$
A {\em quaternionic code} $\cC$ is a non-empty subgroup of $\cQ_8^{n}$.
Quaternionic codes can also be seen as binary codes under the following Gray map:
$\phi:\cQ_8\;\longrightarrow\;\Z_2^4$, such
that
\begin{eqnarray*}
\phi(\one)=(0,0,0,0),    & \phi(\bb)=(0,1,1,0), \\
\phi(\ba)=(0,1,0,1),   & \phi(\ba\bb)=(1,1,0,0), \\
\phi(\ba^2)=(1,1,1,1), & \phi(\ba^2\bb)=(1,0,0,1), \\
\phi(\ba^3)=(1,0,1,0), & \phi(\ba^3\bb)=(0,0,1,1).
\end{eqnarray*}
We will also denote by $\phi$ the componentwise extended map from
$\cQ_8^n$ to $\Z_2^{4n}$.
If $\cC$ is a quaternionic code, then we will say that $C=\phi(\cC)$ is a $\cQ_8$-{\em code} of binary length $4n$.

Binary linear codes, quaternary linear codes and $\cQ_8$-codes can be seen as particular cases of a more general family of codes.
More specifically, given non-negative integers $k_1,k_2$ and $k_3$ we can define the generalized Gray map
$$
\Phi:\;\Z_2^{k_1}\times\Z_4^{k_2}\times\cQ_8^{k_3}\;\longrightarrow\;\Z_2^{k_1+2k_2+4k_3},
$$
such that if $x\in\Z_2^{k_1}$, $y\in\Z_4^{k_2}$ and $z\in\cQ_8^{k_3}$ then
$$
\Phi(x,y,z)=(x,\varphi(y),\phi(z)).
$$
A $\Z_2\Z_4\cQ_8$-{\em code} is a binary code $C$ of the form $C=\Phi(\cC)$ where $\cC$ is a subgroup of
$\Z_2^{k_1}\times\Z_4^{k_2}\times\cQ_8^{k_3}$.

Notice that if $k_1>0$ and $k_2=k_3=0$ then $C$ is a binary linear or
$\Z_2$-linear code. If $k_2>0$ and $k_1=k_3=0$, then $C$ is a
$\Z_4$-linear code. If $k_3>0$ and $k_1=k_2=0$, then $C$ is a
$\cQ_8$-code. Finally, if $k_3=0$, then $C$ is called a
\zz code \cite{Z2Z4} (hence, including the cases
$\Z_2$-linear and $\Z_4$-linear). We also remark that $\cC$ is abelian if and only if $C$ is a
\zz code.

We use additive notation for $\Z_2$ and $\Z_4$ and multiplicative notation for $\cQ_8$ and $\Z_2^{k_1}\times \Z_4^{k_2}\times \cQ_8^{k_3}$.
Therefore, the identity of $\Z_2^{k_1}\times \Z_4^{k_2}\times \cQ_8^{k_3}$ is $(0,\stackrel{k_1+k_2}{\dots},0,\one,\stackrel{k_3}{\dots},\one)$.
We denote this element as $\be$. We also denote it $\be_{k_1,k_2,k_3}$ or $\be_l$, with $l=k_1+k_2+k_3$, when we want to emphasize the ambient space of $\be$.
Note that each of the groups $\Z_2$, $\Z_4$ and $\cQ_8$ have exactly one element of order 2. So there is a unique element $\bu$ of $\cG$ which has the element of order 2 in each coordinate. This element is also determined by the fact that $\Phi(\bu)$ is the all one vector. As for $\be$, we also denote $\bu$ by $\bu_{k_1,k_2,k_3}$ or $\bu_l$, with $l=k_1+k_2+k_3$ if we want to emphasize its ambient space.

If $h\in \Z$ and $w=(x,y,z)\in\Z_2^{k_1}\times\Z_4^{k_2}\times \cQ_8^{k_3}$, where
$x\in\Z_2^{k_1}$, $y\in\Z_4^{k_2}$ and $z\in \cQ_8^{k_3}$, then
	$$w^h = (hx,hy,z^h).$$
The order of $w$ is the smallest positive integer $h$ such that $w^h=\be$.

Let ${\cal S}_n$ denote the symmetric group of permutations on the
set $\{1,\ldots,n\}$. For any $\pi\in {\cal S}_n$ and any vector
$v=(v_1,\ldots,v_n)\in \Z_2^n$, we write $\pi(v)$ to denote the
vector $(v_{\pi^{-1}(1)},\ldots,v_{\pi^{-1}(n)})$.

Two binary codes $C_1$ and $C_2$ of length $n$ are said to
be {\em isomorphic} if there is a coordinate  permutation $\pi
\in {\cal S}_n$ such that $C_2=\{ \pi(x) :  x\in C_1 \}$.  They are said
to be {\em equivalent} if there is a vector $y\in \Z_2^n$ and a
coordinate permutation $\pi \in {\cal S}_n$ such that $C_2=\{
y+\pi(x) : x\in C_1 \}$. Although the two definitions above stand
for two different concepts, it follows that two binary linear
codes are equivalent if and only if they are isomorphic.

A binary code $C$ of length $n$ is said to be {\em propelinear}
\cite{rbh} if for any codeword $x\in C$ there exists $\pi_x\in
{\cal S}_n$ satisfying the following properties for all $v\in \Z_2^n$ and $x,y\in C$:
\begin{enumerate}
\item $x+\pi_x(y)\in C$ and
\item $\pi_x(\pi_y(v))=\pi_z(v)$, where $z=x+\pi_x(y)$.
\end{enumerate}

Let $C$ be a propelinear code and for every $x\in C$ let $\pi_x\in S_n$ satisfy the above conditions. For all $x\in C$ and $y\in \Z_2^n$ let $xy=x+\pi_x(y)$. This endows $C$ with a group structure, which is not abelian in general. Therefore, the vector $\zero$ is
always a codeword and $\pi_{\zero}$ is the identity permutation
$I$. Hence, $\zero$ is the identity element in $C$ and
$x^{-1}=\pi_x^{-1}(x)$ for all $x\in C$ \cite{rbh}.
Notice that a binary code may have several structures of propelinear code with different group structures.

%Clearly, any binary linear code is a propelinear code.

The following lemma is straightforward \cite{rbh,rp,br}.

\begin{lemm}\label{DistInv}
Let $C$ be a propelinear code of length $n$. Then,
%$$
%d(u,v)=d(x*u,x*v)\;\;\;for all x\in C\;\; and for all u,v\in \Z_2^n.
%$$
\begin{center}
$d(u,v)=d(xu,xv)$ for all $x\in C$ and $u,v\in \Z_2^n$.
\end{center}
\end{lemm}

This means that  left multiplication in a propelinear code is {\em Hamming compatible}
\cite{bfr} in the sense that $d(xz,x)=\w(z)$ for all $x\in C$ and
$z\in\Z_2^n$.

\begin{defi}
A propelinear code $C$ of length $n$ is said to be {\em
translation invariant} if
\begin{center}
$d(x,y)=d(xu,yu)$ for all $x,y\in C$ and $u\in \Z_2^n$.
\end{center}
\end{defi}

In \cite{rp}, it is proven that a binary code is translation invariant
propelinear if and only if it is a $\Z_2\Z_4\cQ_8$-code.
Then, a translation invariant propelinear binary code is isomorphic to $C=\Phi(\cC)$ for a subgroup of $\cG=\Z_2^{k_1}\times\Z_4^{k_2}\times\cQ_8^{k_3}$. The permutation $\pi_x$ associated to each element of $\cG$ is obtained by concatenation of permutations in each $\Z_4$ or $\cQ_8$ block, such that the permutation in a component of order 2 is the identity; the permutation of a $\Z_4$-coordinate of order 4 is the transposition of the binary components and of a $\cQ_8$-coordinate of order 4 is a product of two disjoint transpositions of the four binary components. More precisely, if $w=(x_1,\dots,x_{k_1},y_1,\dots,y_{k_2},z_1,\dots,z_{k_3})$ and $w'=\Phi(w)$
then $\pi_{w'}=\sigma_1\dots\sigma_{k_2}\delta_1\dots \delta_{k_3}$ where
	\begin{eqnarray*}
	 \sigma_i &=& \left\{\matriz{{ll} I, & \text{if } y_i \in \{0,2\}; \\ (k_1+2i-1,k_1+2i), & \text{if } y_i \in \{1,3\}.}\right.
	\end{eqnarray*}
and if $t=k_1+2k_2$ then
  \begin{eqnarray*}
	 \delta_i &=& \left\{\matriz{{ll} I, & \text{if } z_i \in \{\one,\ba^2\}; \\
																 (t+4i-3,t+4i-2)(t+4i-1,t+4i), & \text{if } z_i \in \{\ba,\ba^3\}; \\
																 (t+4i-3,t+4i-1)(t+4i-2,t+4i), & \text{if } z_i \in \{\bb,\ba^2\bb\}; \\
																 (t+4i-3,t+4i)(t+4i-2,t+4i-1), & \text{if } z_i \in \{\ba\bb,\ba^3\bb\}.}\right.
	\end{eqnarray*}

%The main aim of this paper is to prove that there are translation
%invariant propelinear codes with a quaternionic part and without any
%possible $\Z_2\Z_4$-linear structure. The paper is organized as
%follows. In Section 2 we give examples of codes with quaternionic
%part but also with an additive $\Z_2\Z_4$-linear structure. In
%Section 3 we prove that there are also quaternionic codes without
%any additive structure, we call them `pure' quaternionic codes.
%Finally, in Section 4, we present some interesting `pure'
%quaternionic codes.

The \emph{rank} of a binary code $C$ is the dimension of the binary vector space generated by its codewords. We denote the rank of $C$ with $r(C)$ or simply $r$.
The {\em kernel} of a binary code $C$ of length $n$ is
$$
K(C)=\{z\in\Z_2^n \mid C+z=C\}.
$$
If $C$ contains the all-zero vector, then $K(C)$ is linear. In that case the dimension of $K(C)$ is denoted with $k(C)$ or simply $k$.
These two parameters, the rank and dimension of the kernel,
can be used to classify binary codes, since if two binary codes have
different ranks or dimensions of the kernel, they are non-equivalent.
Note that if $C$ is a propelinear code and $x\in C$ is such that $\pi_x=I$ then $x\in K(C)$.

The binary code $C$ can be partitioned by $K(C)$-cosets and therefore $|C|$ is a multiple of $|K(C)|$. Since the union of $K(C)$
and anyone of its cosets is again linear, it is clear that either $C$ is linear or $|C| > 2|K(C)|$.

If $C$ is not linear and $\overline{C}$ is the linear span of $C$ then $|K(C)|$ divides $|\overline{C}|$ and $|\overline{C}|>|C|$. Therefore, $|\overline{C}|\ge 4|K(C)|$ and $r=\log_2(|\overline{C}|) \ge \log_2(4|K(C)|) = k+2$.
If moreover $C$ is  a $\Z_2\Z_4\cQ_8$-code then $|C|$ is a power of $2$ and therefore, if $C$ is not linear then $|C|\ge 4|K(C)|$. Hence, $|\overline{C}|\ge 8 |K(C)|$ and $r\ge k+3$.
We summarize this in the following lemma.

\begin{lemm}\label{lem:KernelRank}
 If $C$ is a non-linear binary code then $r(C)\ge k(C)+2$. If moreover $C$ is  a $\Z_2\Z_4\cQ_8$-code then $r(C)\ge k(C)+3$.
\end{lemm}

A \emph{Hadamard matrix} of order $n$ is a matrix of size $n \times n$ with entries $\pm 1$, such that $HH^T=n I$. We can easily see that any two rows (columns) of a Hadamard matrix agree in precisely $n/2$ coordinates. If $n>2$ then any three rows (columns) agree in precisely $n/4$ coordinates. Thus, if $n>2$ and there is a Hadamard matrix of orden $n$ then $n$ is multiple of 4. It is conjectured that the converse holds, i.e., if $n$ es multiple of 4 then there are Hadamard matrices of order $n$ \cite{Key}.

Two \emph{Hadamard matrices} are \emph{equivalent} if one can be obtained from the other by permuting rows and/or columns and multiplying rows and/or columns by $-1$. With the last operations we can change the first
row and column of $H$ into $+1$'s and we obtain an equivalent Hadamard matrix which is called normalized. If $+1$'s are replaced by $0$'s and $-1$'s by $1$'s, the initial Hadamard matrix is changed into a (binary) Hadamard matrix and, from now on, we will refer to it when we deal with Hadamard matrices. The binary code consisting of the rows of a (binary) Hadamard matrix and their complements is called a (binary) \emph{Hadamard code}, which is of length $n$, with $2n$ codewords, and minimum distance $n/2$.

\section{Properties of $\Z_2\Z_4\cQ_8$-codes. Rank and dimension of the kernel.}\label{sec:Z2Z4Q8}

In this section we study some of the group theoretical properties of  $\Z_2\Z_4\cQ_8$-codes. We also present an example of a pure $\cQ_8$-code, i.e., a $\cQ_8$-code which is not equivalent to a \zz code.

Throughout this section $\cG = \Z_2^{k_1}\times \Z_4^{k_2}\times \cQ_8^{k_3}$, we fix a non-trivial subgroup $\cC$ of $\cG$ and let $C=\Phi(\cC)$. Then, the length of $C$ is $n=k_1+2k_2+4k_3$ and we set $l=k_1+k_2+k_3$.

We use the notation bellow for $x,y \in \cG$:
\begin{eqnarray*}
x^y & = & y\inv x y, \text{ conjugate of } x \in \cC \text{ by } y \in \cC, \\
(x,y)&=&x\inv y\inv x y, \text{ commutator of } x,y\in \cC, \\
\cC' & =& \gen{(x,y):x,y\in \cC}, \text{ commutator subgroup of } \cC,\\
Z(\cC) &=& \{z\in \cC: zx=xz, \text{ for every } x\in \cC\}, \text{ the center of } \cC,\\
T(\cC) &=& \{z\in \cC : z^2=\be\}.
\end{eqnarray*}

Note that
	$$\cC' \subseteq T(\cC) \subseteq Z(\cC)$$
and hence, both $\cC'$ and $T(\cC)$ are central subgroups of $\cC$. This implies that
	\begin{equation}\label{prop:commutator}
	 (x,y)=(y,x) \text{ and } (xy,z)=(x,z)(y,z)
	\end{equation}
for every $x,y,z\in \cC$.

\begin{defi}\label{def:type}
 We say that $\cC$ is of type $(\sigma,\delta,\rho)$ if $|T(\cC)|=2^{\sigma}$, $[Z(\cC):T(\cC)]=2^{\delta}$ and $[\cC:Z(\cC)]=2^{\rho}$.
\end{defi}

For instance, if $\delta=\rho=0$ then $C$ is a linear code and if $C$ is a \zz code then $\cC\cong \Z_2^{\gamma}\times \Z_4^{\delta}$ for some non-negative integers $\gamma$ and $\delta$. In the latter case $\sigma=\gamma+\delta$ and $\rho=0$. Note that the type depends on the group $\cC$ rather than on the binary code $C$. For example, if  $\cC=\Z_4$ then the type of $\cC$ is $(1,1,0)$. However the corresponding binary code is  $\Z_2^2$ and henceforth, it is also the binary code of a subgroup of $\Z_2^2$ of type $(2,0,0)$. Similarly, if $\cC=\cQ_8$ then $\cC$ has type $(1,0,2)$ but $C$ is linear and hence it is also the binary code of a group of type $(3,0,0)$.

Assume that $\cC$ is of type $(\sigma,\delta,\rho)$.
Clearly $T(\cC)\cong \Z_2^{\sigma}$.
Moreover, as every element of $\cG$ has order 1, 2 or 4, $x^2 \in T(\cC)$ for every $x\in \cC$ and therefore $\cC/T(\cC)\cong \Z_2^{\delta+\rho}$, $Z(\cC)/T(\cC)\cong \Z_2^{\delta}$ and $\cC/Z(\cC)\cong \Z_2^{\rho}$.

Furthermore, $\sigma\ge \delta$, $\cC$ is generated by $\sigma +\delta+\rho$ elements and $x_1,\ldots,x_{\sigma}$; $y_1,\ldots,y_{\delta}$; $z_1,\ldots,z_{\rho}$ with
	$$T(\cC)=\gen{x_1}\times \dots \times \gen{x_{\sigma}} \quad  \text{and} \quad Z(\cC)=\gen{x_1,\dots,x_{\sigma},y_1,\dots,y_{\delta}}\cong \Z_2^{\sigma-\delta}\times \Z_4^{\delta}.$$
Any element in $\cC$ can be written in a unique way as $\prod_{i}x_i^{\alpha_i}\prod_{j}y_j^{\beta_j}\prod_{k}z_k^{\gamma_j}$, where $\alpha_i,\beta_j,\gamma_k\in \{0,1\}$. Moreover this element belongs to $T(\cC)$ if and only if each $\beta_i$ and each $\gamma_j$ are even; and it belongs to $Z(C)$ if and only if each $\gamma_i$ is even.
In particular, the $y_i$'s and $z_i$'s have order 4.

In the remainder of the paper when we write
\begin{equation}\label{generating}
\cC=\gen{\mbox{\genset}}
\end{equation}
we are implicitly assuming that \genset is a generating set of $\cC$ satisfying the above conditions.

\begin{lemm}\label{lemm:cuadrados}
Let $a,b\in \cC\setminus T(\cC)$.
\begin{enumerate}
\item\label{Conmutan} If $(a,b)=\be$ then either $ab\in T(\cC)$ or $a^2\not=b^2$.
\item\label{Tres} $a^2,b^2$ and $(a,b)$ coincide in each non-trivial coordinate of $(a,b)$. In particular, $\w(\Phi((a,b))\le \w(\Phi(a^2))$.
\item\label{sigmatau1} If $\cC$ is of type $(\sigma,\delta,\rho)$ then $\sigma\ge \delta+\min\{1,\rho\}$.
%\item If $\Phi(\cC)$ is a Hadamard code and $a^2=b^2\ne \bu$ then $(a,b)=a^2$.
%Moreover, if $b'\not\in \gen{a,b,T(\cC)}$ then $b'^2\ne a^2$.
\end{enumerate}
\end{lemm}

\begin{proof}~
\ref{Conmutan}. Assume that $a^2=b^2$ and $(a,b)=\be$. As $a^4=\be$, we have $b^2=a^2=a^{-2}$ and hence $(ab)^2=a^2b^2=\be$.

\ref{Tres}. Let $a=(a_1,\dots,a_l)$ and $b=(b_1,\dots,b_l)$.
If the commutator $(a_i,b_i)$ is non-trivial then $a_i$ and $b_i$ are two non-commuting elements of $\cQ_8$ and hence $(a_i,b_i)=a_i^2=b_i^2$.

\ref{sigmatau1}. We know $\sigma \geq \delta$. If $\rho\ne 0$ and $z_1^2 \in \gen{y_1^{2},\dots, y_{\delta}^{2}}$ then, by item~\ref{Conmutan}, $y_1^{\alpha_1}\dots y_{\delta}^{\alpha_{\delta}}z_1 \in T(\cC)$ for some integers $\alpha_1,\dots,\alpha_{\delta}$, contradicting the construction of the generating set. Thus $\gen{y_1^2,\dots, y_{\delta}^2,z^2_1}$ is a subgroup of $T(\cC)$ isomorphic to $\Z_2^{\delta+1}$ and hence $\sigma\ge \delta+1$.
\end{proof}

\begin{defi}
The \emph{swapper} of $x,y\in \cG$ is
$$[x,y] = \Phi\inv(\Phi(x)+\Phi(y)+\Phi(xy)).$$
We will define the swapper of $\cC$ as the set, $$S(\cC) = \{[x,y] : x,y \in \cC\}.$$
\end{defi}

Note that if $x=(x_1,\dots,x_l),y=(y_1,\dots,y_l)\in \cC$ then
	\begin{equation}\label{SwapperComponents}
    [x,y] = ([x_1,y_1],\dots,[x_l,y_l]).
    \end{equation}
Therefore, to compute the swapper it is enough to compute the swapper in $\Z_2$, $\Z_4$ and $\cQ_8$. Clearly $[x,y]=\be$ for $x,y\in \Z_2$. The following tables describe the swapper in $\Z_4$ and $\cQ_8$:
\begin{center}
\begin{table}[h]
\label{Swappers}
\hspace{1cm}
\begin{tabular}{c|cc}
& 0,2 & 1,3 \\\hline
0,2 & 0 & 0 \\
1,3 & 0 & 2
\end{tabular}
\hspace{1cm}
\begin{tabular}{c|cccc}
& $\one$,$\ba^2$ & $\ba$,$\ba^3$ & $\bb$,$\ba^2\bb$ & $\ba\bb$,$\ba^3\bb$ \\\hline
$\one$,$\ba^2$ & $\one$ & $\one$ & $\one$ & $\one$\\
$\ba$,$\ba^3$  & $\one$ & $\ba^2$ & $\ba^2$ & $\one$\\
$\bb$,$\ba^2\bb$ & $\one$ & $\one$ & $\ba^2$ & $\ba^2$\\
$\ba\bb$,$\ba^3\bb$ & $\one$ & $\ba^2$ & $\one$ & $\ba^2$
\end{tabular}
\caption{Swappers}\label{table:1}
\end{table}
\end{center}

In particular $[x,y]\in T(\cG)$ and this implies that
	\begin{equation}\label{SwapperProperty}
	 \Phi([x,y]xy)=\Phi([x,y])+\Phi(xy) = \Phi(x)+\Phi(y),
	\end{equation}
% So
% 	$$xy [x,y] =\Phi\inv(\Phi(x)+\Phi(y)),$$
i.e., the swapper of $x$ and $y$ is the element needed to pass from $\Phi(xy)$ to $\Phi(x)+\Phi(y)$.

Using (\ref{SwapperComponents}) and the swapper tables of $\Z_4$ and $\cQ_8$ one can easily prove the following properties about swappers.

\begin{lemm} \label{lem:prop} Let $x,y,z,t\in \cG$.
\renewcommand{\labelenumi}{\alph{enumi}$)$ }
 \begin{enumerate}
  \item If $z^2=\be$ then $[zx,y]=[x,zy]=[x,y]$ and $[z,x]=[x,z]=\be$.
  \item $[x,x^{-1}]=[x,x]$.
 \item $[x,y][y,x] = (x,y)$.
 \item  $[x,x]=x^2$.
 \item $[x,yz]=[x,y][x,z]$ and $[xy,z]=[x,z][y,z]$.
  \end{enumerate}
\end{lemm}

By Lemma~\ref{lem:prop}, for every $x\in \cG$ the maps
    $$\matriz{{rcccrcc} [x,-]:\cG& \rightarrow & T(\cG)  &       & [-,x]:\cG& \rightarrow & T(\cG) \\
                                y& \mapsto     & [x,y] & \quad &        y & \mapsto & [y,x]}$$
are group homomorphisms and their kernels contain $T(\cG)$. For a subgroup $\cC$ of $\cG$ let
    $$K(\cC) = \{x \in \cC : [x,y]\in \cC \text{ for every } y\in \cC\} = \{x \in \cC : [y,x]\in \cC \text{ for every } y\in \cC\}.$$
The equality of the above two sets is a consequence of item \emph{c)} in Lemma~\ref{lem:prop}. Moreover by item~\emph{e)} in Lemma~\ref{lem:prop}, $K(\cC)$ is a subgroup of $\cC$.

The following lemma is a consequence of the definition of swapper and the fact that $T(\cC)\subseteq \ker [x,-]$ for every $x$.

\begin{lemm}[Lower bound for $k(C)$]\label{lem:k}
Let $\cC$ be a $\Z_2\Z_4\cQ_8$-code of type $(\sigma,\delta,\rho)$. Then $\Phi(T(\cC)) \subseteq \Phi(K(\cC))= K(\Phi(C))$ and hence $\delta\le \sigma\le k(C)$.
\end{lemm}

Whenever we have a quotient group $G/N$, with $G$ a group and $N$ a normal subgroup of $G$, and the meaning be clear from the context we use the standard bar notation, i.e., if $g\in G$ then $\overline{g}=gN$, the $N$-coset containing $g$.

\begin{lemm}[Upper bound for $r(C)$]\label{lem:KernelSpan}
Let $\cC$ be a $\Z_2\Z_4\cQ_8$-code of type $(\sigma,\delta,\rho)$. Let $\cD=\gen{\cC\cup S(\cC)}$, the group generated by $\cC$ and the swappers of the elements of $\cC$.
Then
\begin{enumerate}
\item $\Phi(\cD)$ is the binary linear span of $C$;
\item if $|\cC|=2^m$ then $r(C)\le m+\binom{m-k(C)}{2}$;
\item $r(C)\le \sigma+\delta+\rho+h$ with $h\le \min\left\{\binom{\delta+\rho}{2},l-\sigma\right\}$.
\end{enumerate}
\end{lemm}

\begin{proof}
By (\ref{SwapperProperty}), it is clear that $\Phi(\cD)$ is included in the linear span of $C$. To prove the inverse it is enough to show that $\Phi(\cD)$ is closed under addition. To see this let $x_1,x_2\in \cD$.
As all the swappers have order at most 2 we have $x_1=b_1c_1$ and $x_2=b_2c_2$ with $b_1,b_2\in \cC$ and $c_1,c_2\in \gen{[c,c']:c,c'\in \cC}$. Then, by (\ref{SwapperProperty}) and item a) of Lemma~\ref{lem:prop} we have
	$$\Phi(x_1)+\Phi(x_2) = \Phi(x_1x_2[x_1,x_2]) = \Phi(b_1b_2c_1c_2[x_1,x_2]) \in \Phi(\cD),$$
as desired.

Let $\cC=\gen{K(\cC),a_1,\dots,a_t}$ with $t$ minimal. Since $T(\cC)\subseteq K(\cC)$, we have $a_i^2\in K(\cC)$ for every $i$ and $t\le \delta+\rho$. Let $c,c'\in \cC$. Then $c=x \prod_{i=1}^t a_i^{\alpha_i}$ and $c'=x'\prod_{i=1}^t a_i^{\beta_i}$ with $x,x'\in K(\cC)$ and each $\alpha_i,\beta_i\in\{0,1\}$. Using Lemma~\ref{lem:prop} we have
	\begin{eqnarray*}
	 [c,c'] &=&
		[x,x']\prod_{i=1}^t [x,a_i]^{\beta_i} [a_i,x']^{\alpha_i}\prod_{1\le i,j \le t} [a_i,a_j]^{\alpha_i\beta_j}\\
					&=& [x,x']\prod_{i=1}^t [x,a_i]^{\beta_i} [a_i,x']^{\alpha_i}a_i^{2\alpha_i\beta_i} \prod_{1\le i<j< t} (a_i,a_j)^{\alpha_j\beta_i}[a_i,a_j]^{\alpha_i\beta_j+\alpha_j\beta_i}.
	\end{eqnarray*}
By Lemma~\ref{lem:k}, $[x,x']\prod_{i=1}^t [x,a_i]^{\beta_i} [a_i,x']^{\alpha_i} a_i^{2\alpha_i\beta_i} \prod_{1\le i<j< h} (a_i,a_j)^{\alpha_j\beta_i}\in \cC$.
This proves $\cD\subseteq \gen{\cC,[a_i,a_j]:1\le i<j\le t}$. The reverse inclusion is obvious, hence $\cD=\gen{\cC\cup \{[a_i,a_j]:1\le i<j\leq t\}}$.

Now, from Lemma~\ref{lem:k} and having in mind that $|S(\cC)|\le |T(\cG)|\le 2^{k_1+k_2+k_3}=2^l$ we have that $\cD$ is of type $(\sigma+h,\delta,\rho)$ and so $r(C)= \sigma+\delta+\rho+h$ with $h\le \binom{t}{2}\le \binom{\delta+\rho}{2}$ and $h\le l-\sigma$.
\end{proof}

\begin{lemm}\label{lineal}
Let $C$ be a \zz code. If $|C|\leq 8$, then $C$ is also a $\Z_2$-linear code.
\end{lemm}

\begin{proof}
 Assume that $\cC$ is of type $(\sigma,\delta,\rho)$. As $\cC$ is commutative $\rho=0$. If $C$ is not linear then $8\ge 2^{\sigma+\delta}=|C|>2|K(C)|=2^{k(C)+1}$. Hence, by Lemma~\ref{lem:k}, we have $\delta\le \sigma \le k(C)\le 1$ and so $\sigma+\delta \leq 2$. Thus, $k(C)+1<\delta+\sigma\le 2$ and so $k(C)=\delta=\sigma=0$, a contradiction. Hence, $C$ is $\Z_2$-linear.
\end{proof}

We now present a ``pure'' $\cQ_8$-code.

\begin{prop} \label{Q8Example}
Let $\cC$ be the quaternionic code $\cC=\gen{(\ba,\ba),
(\ba\bb,\bb)} \leq \cQ_8^2$. Let $C=\Phi(\cC)$ be the corresponding $\cQ_8$-code. Then, $C$ is not a \zz code.
\end{prop}

\begin{proof}
$\cC$ has eight elements, namely
\begin{eqnarray*}
\cC =& \{(\one ,\one ),(\ba,\ba),(\ba^2,\ba^2),(\ba^3,\ba^3),(\ba\bb,\bb),(\ba^2\bb,\ba\bb),(\ba^3\bb,\ba^2\bb),(\bb,\ba^3\bb)\}.
\end{eqnarray*}
The swapper $[(\ba,\ba),(\ba\bb,\bb)]=(\ba^2,\one)\not\in \cC$ and hence, by Lemma~\ref{lem:KernelSpan}, $C$ is not linear.
Finally, by Lemma~\ref{lineal}, $C$ is not \zz.
\end{proof}

\begin{rema}
\rm  The type of the group $\cC$ of Proposition~\ref{Q8Example} is $(\sigma,\delta,\rho)=(1,0,2)$. Let $C=\Phi(\cC)$ and let $k=k(C)$ and $r=r(C)$. By Lemma~\ref{lem:KernelRank} we have $r\ge k+3$, by Lemma~\ref{lem:k}, we have $k\ge \sigma=1$; by Lemma~\ref{lem:KernelSpan}, we have $r\le \sigma+\delta+\rho+\binom{\delta+\rho}{2}=1+2+1=4$ and by statement~\ref{sigmatau1} of Lemma~\ref{lemm:cuadrados}, $\sigma\ge \delta+\min\{1,\rho\}=1$. In this example the previous bounds on $\sigma$, the rank and the dimension of the kernel are absolutely tight, we have $\sigma=1$, $k=1$ and  $r=4$.\end{rema}

\section{Hadamard $\Z_2\Z_4\cQ_8$-codes}\label{sec:Hadamard}

In this section we will focus on Hadamard $\Z_2\Z_4\cQ_8$-codes. But, first of all, we shall begin seeing that the usual companions of the Hadamard codes, that is the (extended) perfect codes which are $\Z_2\Z_4\cQ_8$-codes, do not exist for binary length $n>8$, except for those which are \zz codes.
However we will present a number of Hadamard $\Z_2\Z_4\cQ_8$-codes. The main result of this section is Theorem~\ref{ClasificacionHadamard} which provides a classification of Hadamard $\qq$-codes in terms of its structure.
For Hadamard $\qq$-codes we also refine the upper bound for the rank given in Lemma~\ref{CuadradosIndependientes} for arbitrary $\qq$-codes.
As an application we classify the Hadamard codes of length 16 which are $\qq$-codes.
More precisely, for any $n=2^m$, there is an unique (up to isomorphism) extended Hamming code of length $n$ and, since the dual of an extended Hamming code is a Hadamard code~\cite{mac}, for the same length $n$ it always exist Hadamard codes. But, there are much more non isomorphic Hadamard codes. As an example, there are exactly five non isomorphic Hadamard codes of length 16~\cite{Key}.
One of them is the linear Hadamard code. The other four have the following parameters for the rank $r$ and the dimension of the kernel $k$:
$(r,k) \in \{ (6,3), (7,2), (8,2), (8,1) \}$ \cite{PRV05}.
The Hadamard code with parameters $(r,k)=(6,3)$ is a \zz code, and the
other three nonlinear Hadamard codes are not \zz \cite{PRV06}.
In Proposition \ref{Q8linearN16}, we show that the one with parameters  $(r,k)=(7,2)$ is
a $\cQ_8$-code. Moreover, we will see that the remaining two codes, the ones with parameters
$(8,2)$ and $(8,1)$, are not $\Z_2\Z_4\cQ_8$-codes (see Example~\ref{Hadamar16}). This is a consequence of an analysis of the structure and relations of the type and parameters of Hadamard $\Z_2\Z_4\cQ_8$-codes.

It is well known that there exist \zz perfect codes and extended perfect Hadamard codes~\cite{br,kr}, but we are interested in (extended) perfect $\Z_2\Z_4\cQ_8$-codes where the quaternion group is involved.

\begin{theo}
 If $n=k_1+2k_2+4k_3$ with $k_3>0$ then there is a perfect (respectively, extended perfect) binary code of the form $\Phi(\cC)$ with $\cC$ a subgroup of $\Z_2^{k_1}\times \Z_4^{k_2}\times \cQ_8^{k_3}$ if and only if $n=7$ and $(k_1,k_2,k_3)=(3,0,1)$ (respectively, either $n=8$ and $(k_1,k_2,k_3)\in \{(4,0,1)$, $(0,2,1),(0,0,2)\}$; or $n=4$ and $(k_1,k_2,k_3)=(0,0,1)$.)
\end{theo}

\begin{proof}
The statement about perfect codes was already proved in \cite{br}.

We begin with the exceptional extended perfect codes. The extended Hamming code of length 8 can be constructed as the binary code of a  $\Z_2\Z_4\cQ_8$-code $\cC$ in the following three ways: taking $\cC$ as the subgroup of $\Z_2^4\times \cQ_8$ generated by  $(1,1,0,0,\ba^3)$, $(1,0,1,0,\ba^2\bb)$ and $(1,1,1,1,\ba^2)$; taking $\cC$ as the subgroup of $\Z_4^2\times \cQ_8$ generated by  $(2,0,\ba^3), (1,3,\ba^2\bb)$ and $(2,2,\ba^2)$; and taking $\cC$ as the subgroup of $\cQ_8^2$ generated by  $(\ba,\ba),(\bb,\bb)$ and $(\one,\ba^2)$. Puncturing the first one in the first coordinate we obtain the Hamming code of length 7 which is the binary code associated to the subgroup of $\Z_2^3\times \cQ_8$ generated by $(1,0,0,\ba^3)$, $(0,1,0,\ba^2\bb)$ and $(1,1,1,\ba^2)$.
The extended Hamming code of length 4 can be constructed as the binary code of the subgroup of $\cQ_8$ generated by  $\ba^2$.

We now prove the main assertion of the statement. Let $\cC$ be a subgroup of $\Z_2^{k_1}\times \Z_4^{k_2}\times \cQ_8^{k_3}$ with $k_3\ne 0$ such that $C=\Phi(\cC)$ is a extended perfect code. Then $n=k_1+2k_2+4k_3=2^m$ for some $m$. All the argument is based on the fact that every vector of weight 3 has to be at Hamming distance 1 of an element of $\Phi(\cC)$ of weight 4.

Assume that $k_3\ge 2$ and $(k_1,k_2,k_3)\ne (0,0,2)$ and take a vector of the form $(x|1,0,0,0|1,0,0,0)$ where the last two blocks of length 4 correspond to two $\cQ_8$-coordinates and $x$ has weight 1. Then this element is not at distance 1 of any codeword. This proves that either $k_3=1$ or $(k_1,k_2,k_3)=(0,0,2)$.
In the latter case $\cC$ is the exceptional subgroup of $\cQ_8^2$.

So in the remainder of the proof we assume that $k_3=1$:

If $k_1=k_2=0$ then $n=4$ and necessarily $\cC=\gen{a^2}$, the unique subgroup of $\cQ_8$ of order 2.

If $k_1 \not=0$ then we can puncture the code at a $\Z_2$-coordinate obtaining a perfect code which should be the binary code of the subgroup of $\Z_2^3\times \cQ_8$ obtained above. Hence, when $k_1\not=0$, $\cC$ is the above exceptional afore mentioned subgroup of $\Z_2^4\times \cQ_8$, which coincides with the unique linear Hadamard code of length 8.

Finally, suppose that $k_1=0$ and $k_2\ne 0$. Then $3\le k_2+2=k_2+2k_3= 2^{m-1}$ and therefore $k_2\equiv 2 \mod 4$.
We claim that $k_2=2$. Otherwise, $k_2 \geq 6$.
For $j=2,\dots,k_2$, let $y_j$ (respectively, $y'_j$) denote the element of $\Z_4^{k_2}$ having 1 at the first and $j$-th entries (respectively, 1 at the first entry and $-1$ at the $j$-th entry) and zero at the other entries.
Then both $(\phi(y_j)|1,0,0,0)$ and $(\phi(y'_j)|1,0,0,0)$ are vectors of weight 3 and hence they are at distance one of some $x_i,x'_i\in \Phi(\cC)$, respectively.
Then $x_i=\Phi(y_j|a_i)$ and $x'_i=\Phi(y'_j|a'_i)$ for some $a_i,a'_i\in \cQ_8$ of order 4. As $\phi(y_i-y_j)$ and $\phi(y'_i-y'_j)$ have both weight 2, the $a_i$'s are pairwise different elements of order 4 and so are the $a'_i$'s. Moreover $\cQ_8$ has exactly 6 elements of order 4.
Thus $a_i=a'_j$ for some $i,j$. Therefore $\Phi(y_i-y'_j|\one)$ is an element of $\Phi(\cC)$ of weight 2, a contradiction.
\end{proof}

In the following, we use the notation $N\rtimes H$ for a semidirect product, i.e., a group such that as a set $N\rtimes H=N\times H$ with multiplication given by
	$$(n_1,h_1)(n_2,h_2)=(n_1\alpha_{h_1}(n_2),h_1h_2) \quad (n_1,n_2\in N, h_1,h_2\in H),$$
where $h\mapsto \alpha_h$ is a group homomorphism $\alpha:H\rightarrow \Aut(N)$.
%An alternative description of $N\rtimes H$ is as follows: It has a normal subgroup $N$ and a subgroup of $H$ such that the map $(n,h)\mapsto nh$ is a bijection $N\times H\rightarrow N\rtimes H$. In this description $\alpha$ takes the form of the action of $H$ on $N$ by conjugation and the product in $N\rtimes  H$ is given by $(n_1h_1)(n_2h_2)=(n_1n_2^{h_1^{-1}})(h_1h_2)$, for $n_1,n_2\in N$ and $h_1,h_2\in H$.
The direct product $N\times H$ is the semidirect product with $\alpha_h=I$ for every $h\in H$.

\begin{prop}\label{Q8linearN16}
Consider the quaternionic code
\begin{equation}\label{h72}
\cC=\gen{(\ba,\ba,\ba,\ba),(\bb,\ba \bb,\bb,\ba \bb),(\ba^2,\one ,\ba,\ba^3)} \leq  \cQ_8^4
\end{equation}
and let $C=\Phi(\cC)$. Then $\cC$ is of type $(2,0,3)$ and $C$ is a Hadamard code of length $16$, rank 7 and dimension of the kernel 2.
\end{prop}

\begin{proof}
Let $a=(\ba,\ba,\ba,\ba),b=(\bb,\ba \bb,\bb,\ba \bb)$ and $c=(\ba^2,\one ,\ba,\ba^3)$. Then $a$, $b$ and $c$ have order $4$, $a^2b^2=\be$, $a^b=a^{-1}$, $c^a=c$ and $c^b=c^{-1}$. Moreover $\gen{c}\cap \gen{a,b}=\{\be\}$. Therefore, $\cC$ is a semidirect product $\gen{c}\rtimes \gen{a,b}\cong \Z_4\rtimes \cQ_8$. Hence, $T(\cC)=Z(\cC)=\gen{a^2,c^2}\cong \Z_2^2$. Thus, $\cC$ has type $(2,0,3)$ and $G/T(\cC)=\gen{\overline{a}}\times \gen{\overline{b}}\times \gen{\overline{c}}$.
Furthermore, a straightforward calculation shows that $K(\cC)=T(\cC)$ and every element of $C$ has weight $0$, $8$ or $16$. Therefore, $C$ is a Hadamard code. The dimension of its kernel is $2$. By Lemma~\ref{lem:KernelSpan} the linear span of $C$ is $\Phi(\gen{\cC,[a,b]=(\ba^2,\one,\ba^2,\one),[a,c]=\be,[b,c]=(\one,\one,\one,\ba^2)})$. Thus, the rank of $C$ is $7$.
\end{proof}

\begin{lemm}\label{lemm:cuadradosHadamard}
Let $\cC$ be a subgroup of $\Z_2^{k_1}\times \Z_4^{k_2}\times \cQ_8^{k_3}$ such that $\Phi(\cC)$ is a Hadamard code.
Let $a,b,c\in \cC\setminus T(\cC)$.
Then
\begin{enumerate}
 \item either $(a,b)\in \gen{a^2}$ or $a^2=\bu$;
 \item if $a^2=b^2=c^2\ne \bu$ then $|\gen{a,b,c,T(\cC)}/T(\cC)|\le 4$
 \item if $a^2=b^2=(a,b)\ne c^2$ then $\gen{[a,c],[b,c],T(\cC)}/T(\cC)|\le 2$.
\end{enumerate}
\end{lemm}
\begin{proof}
 We know that the binary all ones vector belongs to any Hadamard code and hence, if $\cC$ is a subgroup of $\cG$ such that $\Phi(C)$ is a Hadamard code then $\bu\in \cC$.

1. If $y\in \cC\setminus \{1,\bu\}$ then $\w(\Phi(y))=\frac{n}{2}$. Thus, by item~\ref{Tres} of Lemma~\ref{lemm:cuadrados}, if $a^2\ne \bu$ and $(a,b)\ne \be$ then $(a,b)=a^2$.

2. Assume $a^2=b^2=c^2=t\ne \bu$ and $|\gen{a,b,c,T(\cC)}/T(\cC)|>4$. Then $\{ab,ac,bc,abc\}\cap T(\cC)=\emptyset$. Therefore, by items~\ref{Conmutan} and \ref{Tres} of Lemma~\ref{lemm:cuadrados}, $(a,b)=(a,c)=(b,c)=t$. Hence $(abc)^2=(a,b)(a,c)(b,c)a^2b^2c^2=t^6=\be$. Thus $abc\in T(\cC)$, a contradiction.

3. Suppose $a^2=b^2=(a,b)\ne c^2$ and $|\gen{[a,c],[b,c],T(\cC)}/T(\cC)|>2$.
Then $\gen{[a,c],[b,c]}$ is isomorphic to $\Z_2^2$ and intersects $T(\cC)$ trivially.
Write $a=(a_1,\dots,a_l), b=(b_1,\dots,b_l)$ and $c=(c_1,\dots,c_l)$.

Suppose $(a,c)=(b,c)=\be$. Then $a^2\ne \bu$ for otherwise $\gen{a_i,b_i}=\cQ_8$ for every $i$ and hence $c$ has order 2, contradicting the fact that $[a,c]\ne \be$.
Therefore, after reordering the coordinates we may assume that $a^2=(\ba_{l_1}|\be_{l_2})$. Hence, if $i\le l_1$ then $\gen{a_i,b_i}=\cQ_8$ and $c_i$ has order at most two and when $i > l_1$, $a_i$ and $b_i$ have order at most two and $c_i$ has order four. Then $[a,c]=\be$, contradicting the assumption.

Thus either $(a,c)\ne \be$ or $(b,c)\ne \be$. By symmetry we may assume that $(a,c)\ne \be$. If $(b,c)=\be$ then $(ab,c)\ne \be$, $(ab)^2=a^2$ and $[ab,c]=[a,c][b,c]$, so that $\gen{[a,c],[b,c]}=\gen{[a,c],[ab,c]}$. Therefore, we can replace $b$ by $ab$ and so we may assume that $(b,c)\ne \be$.
Then, by item 1, either $a^2=\bu$ and $(a,c)=(b,c)=c^2$ or $c^2=\bu$ and $(a,c)=(b,c)=a^2$.

Suppose that $a^2=\bu$. Then $\gen{a_i,b_i}\cong \cQ_8$ for every $i$, so that $\cG=\cQ_8^l$ and $\w(c^2)=2l=\frac{n}{2}$. Then $l$ is even and after reordering the coordinates we may assume that $c^2=(\bu_{\frac{l}{2}}|\be_{\frac{l}{2}})$. Then each $a_i$ and $b_i$ have order 4 and $c_i$ has order 4 if and only if $i\le \frac{l}{2}$. Let $A_1=\{\ba,\ba^3\}$, $A_2=\{\bb,\ba^2\bb\}$ and $A_3=\{\ba\bb,\ba^3\bb\}$.
If $r\in A_i$ and $s\in A_j$ then $(r,s)=\be$ if and only if $i=j$; and $[r,s]=\be$ if and only if $i-j\equiv 1 \mod 3$ (see Table~\ref{table:1}).
Each $a_i$, $b_i$  and $c_i$,  with $i\le \frac{l}{2}$, belongs to some $A_i$ and $(a_i,b_i)=(a_i,c_i)=(b_i,c_i)=\ba^2$. Therefore $a_i$, $b_i$ and $c_i$ belong to different $A_i$'s. This implies that for every $i\le \frac{l}{2}$, $\{[a_i,c_i],[b_i,c_i]\}=\{1,\ba^2\}$. On the other hand, $(a_i,c_i)=(b_i,c_i)=\one$ for every $i>\frac{l}{2}$. Therefore $[a,c][b,c]=(\bu_{\frac{l}{2}}|\be_{\frac{l}{2}})=c^2\in T(\cC)$ which yields a contradiction. A slight modification of this argument, with the roles of $a$ and $c$ interchanged, yields also a contradiction in the case $c^2=\bu$. This finishes the proof of 3.
\end{proof}

The following corollary is a straightforward consequence of Lemma~\ref{lemm:cuadrados} and Lemma~\ref{lemm:cuadradosHadamard}.

\begin{coro}\label{HadamardZs}
Let $\cC$ be a subgroup of $\Z_2^{k_1}\times \Z_4^{k_2}\times \cQ_8^{k_3}$ such that $\Phi(\cC)$ is a Hadamard code and let \genset be a set of generators of $\cC$.
\begin{enumerate}
 \item For every $t\in T(\cC)$, $t\not=\bu$, the cardinality of $\{i=1,\dots,\rho:z_i^2=t\}$ is at most 2.
 \item \label{segon} If $(z_i,z_j)\ne \be$ then either $z_i^2=\bu$ or $(z_i,z_j)=z_i^2$.
 \item If $(z_i,z_j)=e$ then $z_i^2\ne z_j^2$.
\end{enumerate}
\end{coro}

Corollary~\ref{HadamardZs} implies that if $\cC$ is a subgroup of $\cG$ such that $\Phi(\cC)$ is a Hadamard code and $t\ne \bu$ then a generating set of $\cC$ has at most two $z_i$'s with square equal to $t$. Moreover if $z_i^2=z_j^2\ne \bu$ then $(z_i,z_j)=z_i^2$, by item~\ref{Conmutan} of Lemma~\ref{lemm:cuadrados} and item 1 of Lemma~\ref{lemm:cuadradosHadamard}. For our proposes it is convenient to use a generating set for which this property also holds for $z_i$'s with square $\bu$.

\begin{defi}
Let $\cC$ be a subgroup of $\Z_2^{k_1}\times \Z_4^{k_2}\times \cQ_8^{k_3}$ such that $\Phi(\cC)$ is a Hadamard code and let \genset be a set of generators of $\cC$. We say that this generating set is \emph{normalized} if $z_i^2=\bu$ for at most two $i=1,\dots,\rho$ and if $z_i^2=z_j^2=\bu$ with $i\ne j$ then $(z_i,z_j)=\bu$.
\end{defi}

\begin{lemm}
Every Hadamard $\Z_2\Z_4\cQ_8$-code $\cC$ has a normalized set of generators.
\end{lemm}
\begin{proof}
Let $x_1,\ldots,x_\sigma; y_1\ldots,y_\tau; z_1,\ldots, z_\rho$ be a generating set of $\cC$. We may assume without loss of generality that $z^2_i=\bu$ if and only if $i\le k$ and either $(z_1,z_i)\ne \bu$ for every $2\le i \le k$ or $(z_1,z_2)=\bu$.
If $k\le 1$ there is nothing to prove. Otherwise, for every $i=1,\dots,\rho$ we define
	$$z'_i = \left\{\matriz{{ll}
						z_i, & \text{if either } i=1, \text{ or } i=2\le k \text{ and } (z_1,z_2)=\bu \text{ or } i > k; \\

						z_1z_i, & \text{if } 2\le i \le k \text{ and } (z_1,z_i)\ne \bu; \\
						z_2z_i, & \text{if } 3\le i \le k, (z_1,z_i)=\bu\ne (z_2,z_i); \\
						z_1z_2z_i, & \text{if } 3\le i \le k, (z_1,z_i)=\bu = (z_2,z_i).}\right.
$$
Then $x_1,\ldots,x_\sigma; y_1\ldots,y_\tau; z'_1,\ldots, z'_\rho$ is a generating set of $\cC$ and we claim that it is normalized.
Indeed, if $i>k$ then $z_i'^2=z_i^2\ne \bu$.
Assume $3\le i \le k$.
If $(z_1,z_i)\ne \bu$ then $z_i'^2=(z_1z_i)^2 = (z_1,z_i)z_1^2z_2^2 = (z_1,z_i)\ne \bu$.
Assume that $(z_1,z_i)=\bu$. Then, by construction, $(z_1,z_2)=\bu$. Therefore, if $(z_2,z_i)\ne \bu$ then $z_i'^2=(z_2z_i)^2=(z_2,z_i)z_2^2z_i^2=(z_2,z_i)\ne \bu$, and if $(z_2,z_i)=\bu$ then $z_i'^2 = (z_1z_2z_i)^2 = (z_1,z_2)(z_1,z_i)(z_2,z_i)z_1^2z_2^2z_i^2 = \bu^6=1\ne \bu$. Finally, assume that $i=2\le k$. If $(z_1,z_2)=\bu$ then $z_2'=z_2$ and so $(z_1,z'_2)=\bu$ and otherwise $z_2'^2=(z_1z_2)^2=(z_1,z_2)\ne \bu$.
\end{proof}

In the remainder of the section
%$\cG=\Z_2^{k_1}\times \Z_4^{k_2}\times \cQ_8^{k_3}$, $l=k_1+k_2+k_3$, $n=k_1+2k_2+4k_k$, $\cC$ is a subgroup of $\cG$ such that $C=\Phi(\cC)$ is a Hadamard code and
\genset is a normalized generating set of $\cC$. Moreover, let $\epsilon$ be the number of pairs of different $z_i$'s with the same squares, We reorder the $z_i$'s in such a way that two $z_i$'s with the same square are consecutive and placed at the beginning of the list, i.e., $z_1^2=z_2^2,\dots,z_{2\epsilon-1}^2=z_{2\epsilon}^2$ and $z_2^2,z_4^2,\dots,z_{2\epsilon}^2,z_{2\epsilon+1}^2,\dots,z_{\rho}^2$ are pairwise different. Note that for each $i=1,\dots,\rho$ there is $j=1,\dots,\rho$ such that $(z_i,z_j)\ne \be$. In that case, either $z_i^2=\bu$ or $z_j^2=\bu$ or $z_i^2=z_j^2$. In the last case $\{i,j\}=\{2t-1,2t\}$ for some $t\le \epsilon$.

\begin{lemm}\label{CuadradosIndependientes}
Let $\cC$ be a subgroup of $\Z_2^{k_1}\times \Z_4^{k_2}\times \cQ_8^{k_3}$ such that $\Phi(\cC)$ is a Hadamard code and let \genset be a normalized set of generators of $\cC$. Then the following assertions hold:
\begin{enumerate}
 \item $\epsilon\le 2$.
 \item If $\epsilon=2$ then $\delta=0$ and $\rho=4$. In the case when $z_1^2,z_3^2$ and $\bu$ are pairwise different we have $(z_i,z_j)=\be$ and $z_i^2z_j^2=\bu$ for $i\in \{1,2\}$, $j\in \{3,4 \}$.
 \item If $z_i^2=\bu$ with $i\le 2\epsilon$ then $\delta=0$.
 \item Let $V=\{y_1,\dots,y_{\delta},z_1,z_3,\dots,z_{2\epsilon-1},z_{2\epsilon+1},z_{2\epsilon+2},\dots,z_{\rho}\}$ and  $W=\{w^2 : w\in V\}$ and $U=\{ u\in V : u^2 \not=\bu \}$. Then $|\gen{W}|\ge 2^{\delta+\rho-\epsilon-1}$, and hence $\sigma\ge \delta+\rho-\epsilon-1$. If moreover $\bu \notin \gen{U}$ then $|\gen{W}|=2^{\delta+\rho-\epsilon}$, and hence $\sigma\ge \delta+\rho-\epsilon$.
 \item {\rm (Upper bound for $r(C)$)} If $\cD=\gen{\cC\cup S(\cC)}$ (as in Lemma~\ref{lem:KernelSpan}) then $\cD$ is of type $(\sigma+h,\delta,\rho)$ and $r(C)\le \sigma+\delta+\rho+h$ with
	$$h\le \left\{\matriz{{ll} \epsilon+\binom{\delta+\rho-\epsilon}{2}, & \text{if } \epsilon\le 1;\\ 3, & \text{if } \epsilon=2}\right.$$
\end{enumerate}
\end{lemm}

\begin{proof}~
Item 3. Assume $z_i^2=\bu$ for some $i\le 2\epsilon$. Then we may assume without loss of generality that $z_1^2=z_2^2=(z_1,z_2)=\bu$. (Recall that our generating set is normalized.) Then the projection of $\gen{z_1,z_2}$ onto each coordinate is $\cQ_8$, so that $k_1=k_2=0$, and no element of $\cC$ of order $4$ is central, i.e., $\delta=0$.

For every $i=1,\dots,\rho$ let $X_i=\{j:\text{ the } j\text{-th coordinate of } z_i \text{ has order 4}\}$.

We claim that if $z_i^2, z_j^2$ and $\bu$ are pairwise different for some $i,j\le 2\epsilon$ then $k_1=k_2=\delta=0$, $X_i$ and $X_j$ form a partition of $\{1,\dots,l\}$, $z_1^2z_3^2=\bu$ and $\{z_1^2,\dots,z_{\rho}^2\}\subseteq \{z_i^2,z_j^2,\bu\}$. For simplicity we also assume that $i=1$ and $j=3$.
We know that $z_1^2=z_2^2$ and $z_3^2=z_4^2$, and, by Lemma~\ref{lemm:cuadrados}, $(z_1,z_2)\ne \be\ne (z_3,z_4)$. Hence, by Corollary~\ref{HadamardZs}, $(z_1,z_2)=z_1^2=z_2^2$ and $(z_3,z_4)^2=z_3^2=z_4^2$,
and the images of these elements by $\Phi$ have weight $\frac{n}{2}$.
This implies that $X_1=X_2$, the projections of $z_1$ and $z_2$ on the coordinates of $X_1$ generates $\cQ_8$ and $|X_1|=\frac{n}{8}$. Similarly $X_3=X_4$, the projections of $z_3$ and $z_4$ on $X_3$ generate $\cQ_8$ and $|X_3|=|X_4|=\frac{n}{8}$. Moreover, by Corollary~\ref{HadamardZs}, $(z_i,z_j)=\be$ for $i=1,2$ and $j=3,4$. Therefore the projection of $z_3$ and $z_4$ on the coordinates of $X_1$ have order 2. This implies that $X_1$ and $X_3$ are disjoint. Therefore $4(|X_1|+|X_3|)=n$ and hence $k_1=k_2=0$. Furthermore, no element of $\cC$ of order $4$ can commute with each $z_i$ with $i=1,2,3,4$. This implies that $\delta=0$ and, by item 1 of Lemma~\ref{lemm:cuadradosHadamard}, $z_k^2=\bu$ for every $k\ne 1,2,3,4$. This finishes the proof of the claim.

Items 1 and 2. Assume first that $z_i^2,z_j^2$ and $\bu$ are pairwise different for some $i,j\le 2\epsilon$. Observe that this holds if $\epsilon>2$.
By the claim, we may assume without loss of generality that $X_1$ is formed by the first $\frac{n}{8}$ coordinates and $X_3$ is formed by the last $\frac{n}{8}$ coordinates. Moreover, if $\rho>5$, then by the claim $z_5^2=\bu=z_1^2z_3^2$.
Then there are $z'_1\in \gen{z_1,z_2}$ and $z'_3\in \gen{z_3,z_4}$ such that $z'_1$ and $z_5$ coincide in the first coordinate and $z'_3$ and $z_5$ coincide in the last coordinate. As all the coordinates of $z_5$ have order 4, and both the last $\frac{l}{2}$ coordinates of $z_1^2$ and the first $\frac{l}{2}$ coordinates of $z_3^2$ are $\be$, we deduce that the first $\frac{l}{2}$ coordinates of $z'_1$ and the last coordinates $\frac{l}{2}$ of $z'_3$ have order 4. Moreover $(z'_1,z_5)\ne {z'_1}^2$ and $(z'_3,z_5)\ne {z'_3}^2$. So,  $z'^2_1z'^2_3 = z_5^2$ and $(z'_1,z_5)=(z'_3,z_5)=\be$. Then $(z'_1z'_3z_5)^2= z'^2_1z'^2_3z_5^2=\be$, contradicting the fact that $z'_1z'_3z_5\not\in T(\cC)$. This finishes the proof of 1, and proves 2 in case $z_1^2\ne \bu\ne z_3^2$.

Suppose that $\epsilon=2$ and $z_1^2=\bu$. As above, we may assume that $\gen{z_3,z_4}$ projects to $\cQ_8$ in the first $\frac{n}{8}$ coordinates and projects to an element of order at most 2 in the remaining coordinates. Suppose $\rho>4$. By Corollary~\ref{HadamardZs}, $(z_3,z_5)=(z_4,z_5)=\be$, so that the projection of $z_5$ in the first $\frac{n}{8}$ coordinates has order at most 2 and the projection on the remaining coordinates has order 4.
Therefore $(z_3z_5)^2=(z_4z_5)^2 = z_3^2z_5^2=z_4^2z_5^2=\bu$.
Moreover, by the same argument as in the previous paragraph, we could take $z'_1 \in \langle z_1,z_2\rangle$ and  $z'_3 \in \langle z_3,z_4\rangle$ such that $(z'_1,z_5) =(z'_3,z_5)=(z'_1,z'_3)=\be$. This implies that $(z'_1z'_3z_5)^2=z'^2_1z'^2_3z_5^2=\be$, contradicting the fact that $z_1z_3z_5\not\in T(\cC)$. This finishes the proof of 2.

Item 4. Let $t=|U|$ and set $U=\{u_1,\dots,u_t\}$. Observe that $t= \delta+\rho-\epsilon-1$ if $\bu\in W$ and otherwise $t=\delta+\rho-\epsilon$.
From item~\ref{segon} of Corollary~\ref{HadamardZs} the elements of $U$ commute. Moreover the map $\Z_2^t\rightarrow \cC/T(\cC)$ given by $(\alpha_1,\dots,\alpha_t)\mapsto u_1^{\alpha_1}\dots u_t^{\alpha_t}T(\cC)$ is injective. Then, by item~\ref{Conmutan} of Lemma~\ref{lemm:cuadrados}, the rule $(\alpha_1,\dots,\alpha_t)\mapsto u_1^{2\alpha_1}\dots u_t^{2\alpha_t}$ defines a bijection $\Z_2^t\rightarrow T(U)$.
Then $|\gen{W}|\ge|T(\gen{U})|= 2^t\ge 2^{\delta+\rho-\epsilon-1}$.
If $\bu \not\in \gen{U}$ then either $\bu\not\in W$ and hence $t=\delta+\rho-\epsilon$, $W=T(U)$ and $|W|=|T(U)|=2^{\delta+\rho-\epsilon}$ or $\gen{W}=\gen{T(U),\bu}$ and then
$|W|=2|T(U)|=2^{\delta+\rho-\epsilon}$. In both cases $|W|=2^{\delta+\rho-\epsilon}$.

Item 5. Using the argumentation in the proof of Lemma~\ref{lem:KernelSpan} the span of $C$ is $\Phi(\cD)$ where $\cD$ is the group generated $\cC$ and the swappers of $A=\{y_1,\dots,y_{\delta},z_1,\dots,z_{\rho}\}$, where we only take one of the two swappers $[a,b]$ or $[b,a]$ for $a\not=b\in A$.
If $i\le \epsilon$ then $|\gen{[z_{2i-1},a],[z_{2i},a],T(\cC)}/T(\cC)|\le 2$ (Lemma~\ref{lemm:cuadradosHadamard}),  for every $a\in A\setminus \{z_{2i-1},z_{2i}\}$. Therefore, in order to generate $\cD$ modulo $\cC$ it is enough to take the swappers of the form
$[z_j,z_k]$ with $2\epsilon < j < k \le \rho$, the swapper $[z_{2i-1},z_{2i}]$ for each $i\le \epsilon$ and one out of the two swappers $[z_{2i-1},z_j]$ or $[z_{2i},z_j]$ for each $i\le \epsilon$ and $2i<j\le \rho$.
This gives a total of $s=\binom{\delta+\rho-2\epsilon}{2}+\epsilon+\sum_{i=1}^{\epsilon} (\delta+\rho-2i)$ swappers.
Thus $r(C)=\sigma+\delta+\rho+h$ with $h\le s$. If $\epsilon =0$ then $s=\binom{\delta+\rho}{2}$.
If $\epsilon=1$ then $\alpha=1+\binom{\delta+\rho-2}{2}+(\delta+\rho-2)=1+\binom{\delta+\rho-1}{2}$. Finally, assume that $\epsilon=2$. Then $\delta=0$, $\rho=4$ and $h\le s=4$.
We claim that in this case we may assume that $z_1^2=\bu$. Otherwise, by item 2 in this Lemma, $(z_i,z_j)=\be$ and $(z_iz_j)^2=\bu$ for every $i=1,2$ and $j=3,4$, and therefore $(z_1z_3)^2=(z_2z_4)^2=\bu$ and $(z_1z_3,z_2z_4)=(z_1,z_2)(z_3,z_4)=z_1^2z_3^2=\bu$. Thus, replacing $z_1$ and $z_2$ by $z_1z_3$ and $z_2z_4$, respectively, we obtain the desired claim.
To finish the proof it remains to prove that $h\le 3$.
In this case $\cD=\gen{\cC,[z_1,z_2],[z_3,z_4],s_1,s_2}$ where $\{s_1,s_2\}=\{[z_1,z_3],[z_2,z_4]\}$ or $\{s_1,s_2\}=\{[z_1,z_4],[z_2,z_3]\}$. After reordering $z_3$ and $z_4$, if necessary, one may assume that $\cD=\gen{\cC,[z_1,z_2],[z_3,z_4],[z_1,z_3],[z_2,z_4]}$. By means of contradiction assume that $h=4$. This means that
\begin{equation}\label{eq:a}
\cA=\gen{[z_1,z_2],[z_3,z_4],[z_1,z_3],[z_2,z_4],T(\cC)}/T(\cC)
\end{equation}
is of order 16.
By $(z_1,z_2)=\bu$ we have $k_1=k_2=0$. After a suitable reordering we may assume that $z_3^2=(\bu_{\frac{l}{2}},\be_{\frac{l}{2}})$. Write $z_1=(a_1,\dots,a_l)$, $z_2=(b_1,\dots,b_l)$, $z_3=(c_1,\dots,c_l)$ and $z_4=(d_1,\dots,d_l)$.
As $(a_1,b_1)=\ba^2$ and $c_1$ has order 4, $c_1$ does not commute with either $a_1$ or $b_1$. If $(a_1,c_1)=\one$ then replacing $z_1$ by $z_1z_2$ we may assume that $(a_1,c_1)=(b_1,c_1)=\ba^2$. We argue similarly if $(b_1,c_1)=\be$ to deduce that we may always assume that $(a_1,c_1)=(b_1,c_1)=\ba^2$. Then $(z_1,z_3)=(z_2,z_3)=z_3^2$.
For every $x\in \cQ_8$ of order $4$ let $l(x)\in \{1,2,3\}$ with $x\in A_{l(x)}$, where $A_1=\{\ba,\ba^3\}$, $A_2=\{\bb,\ba^2\bb\}$ and $A_3=\{\ba\bb,\ba^3\bb\}$. Then $\{l(a_i),l(b_i),l(c_i)\}=\{1,2,3\}$ for every $i\le \frac{l}{2}$. As $(c_1,d_1)\ne \one$, $l(c_1)\ne l(d_1)$. If $l(b_1)=l(d_1)$ then $(b_1,d_1)=\one$ and therefore $(z_2,z_4)=\be$. Then $l(b_i)=l(d_i)$ for every $i\le \frac{l}{2}$ and hence $s_4=[z_2,z_4]=(\bu_{\frac{l}{2}},\be_{\frac{l}{2}})=z_3^2\in \cC$, contradicting the assumption. Thus $l(b_1)\ne l(d_1)$ and hence $(z_2,z_4)=z_3^2$. Then $l(b_i)\ne l(d_i)$ for every $i\le \frac{l}{2}$. Thus for every $i\le \frac{l}{2}$ we have $\{l(b_i),l(c_i),l(d_i)\}=\{1,2,3\}$, hence $l(a_i)=l(d_i)$ and $[a_i,c_i][c_i,d_i]=[a_i,c_i][d_i,c_i](c_i,d_i)=[a_id_i,c_i](c_i,d_i)=(c_i,d_i)$. We conclude with $[z_1,z_3][z_3,z_4]=(z_3,z_4)^2\in T(\cC)$, a contradiction.
\end{proof}

Next theorem describes the group structure of Hadamard $\Z_2\Z_4\cQ_8$-codes and specify the bounds for the rank and dimension of the kernel in each case.

\begin{theo}\label{ClasificacionHadamard}
 Let $\cC$ is a subgroup of $\Z_2^{k_1}\times \Z_4^{k_2}\times \cQ_8^{k_3}$ such that $\Phi(\cC)$ is a Hadamard code of length $n=2^m$. Let $r=r(C)$ be the rank of $C$ and $k=k(C)$ be the dimension of its kernel.
 Then $\cC$ has a normalized generating set \genset, where $m = \sigma+\delta+\rho-1$, satisfying one of the following conditions.
\begin{enumerate}
 \item\label{rho=0} $\rho=0$. Then $\cC\cong \Z_2^{\sigma-\delta}\times \Z_4^{\delta}$ and $C$ is a \zz code (which could be $\Z_4$-linear code or not), with length $n=2^{m}$ and $m=\sigma+\delta-1$.
 \begin{enumerate}
 \item If $C$ is $\Z_4$-linear then \\
for $\delta\in \{1,2\}$, code $C$ is linear;\\
for $\delta \geq 3$: $k=\sigma+1$; $r = \sigma+\delta+\binom{\delta-1}{2}$.
\item If $C$ is not $\Z_4$-linear then $\sigma > \delta$ and\\
for  $\delta\in \{0,1\}$, code $C$ is linear;\\
for $\delta \geq 2$: $k=\sigma$; $r= \sigma+\delta+\binom{\delta}{2}$.
 \end{enumerate}

% \item\label{rho=2} $\rho=2$, $z_1^2=z_2^2=(z_1,z_2) \not=\bu$. Then $\cC\cong \Z_2^{\sigma-\delta-1}\times \Z_4^{\delta}\times \cQ_8$, $k\ge \sigma\ge \delta+1$ and $r\le \sigma+\delta+3+\binom{\delta+1}{2}=\sigma+2+\binom{\delta+2}{2}$.
 \item\label{rho=arb-tau=0} $\delta=0$, $z_1^2=z_2^2=(z_1,z_2)=\bu$, and $(z_i,z_j)=z_j^2$ and $(z_j,z_k)=\be$ for every $i\in\{1,2\}$ and $3\le j,k \le \rho$. Then $\cC\cong \Z_2^{\sigma-\rho+1}\times (\Z_4^{\rho-2}\rtimes \cQ_8)$,
	$k\ge \sigma\ge \rho-1$ and $r\le \sigma+\rho+1+\binom{\rho-1}{2}$.
 \item\label{rho=arb-tau=arb} $\delta=0$, $z_1^2=\bu\not\in \gen{z_2^2,\dots,z_{\rho}^2}\cong \Z_2^{\rho-1}$ and $(z_1,z_i)=z_i^2$; $(z_i,z_j)=\be$, for every $i\not=j$ in $\{2,\dots,\rho\}$. Then $\cC\cong \Z_2^{\sigma-\rho} \times (\Z_4^{\rho-1}\rtimes \Z_4)$, $k\ge \sigma\ge \rho$ and $r\le \sigma+\rho+\binom{\rho}{2}$.
\item\label{rho=2} $\rho=2$, $\delta \leq 1$, $z_1^2=z_2^2=(z_1,z_2) \not=\bu$. Then $\cC\cong \Z_2^{\sigma-\delta-1}\times \Z_4^{\delta}\times \cQ_8$, $k\ge \sigma\ge \delta+1$ and $r\le \sigma+\delta+\rho+1\leq\sigma+4$.
 \item\label{rho=4} $\delta=0$, $\rho=4$, $z_1^2=z_2^2=(z_1,z_2)=\bu\ne z_3^2=z_4^2=(z_3,z_4)$. Moreover, $(z_i,z_j) \in \langle z_j^2\rangle$ for every $i\in \{1,2\}$ and $j\in\{3,4\}$. Then $\cC\cong \Z_2^{\sigma-2}\times (\cQ_8\rtimes \cQ_8)$ and
	$k\ge \sigma\ge 2$; $r\le \sigma+7$.

\end{enumerate}
\end{theo}

\begin{proof}
If $\cC$ is abelian then condition~\ref{rho=0} holds and the values for the rank and dimension of the kernel are already known \cite{PRV06}. So, in the remainder of the proof we assume that $\cC$ is non-abelian and therefore $\rho\ge 2$. In each case the description of the structure of $\cC$ is a consequence of the relations and the bounds for $r$ and $k$ follow from Lemma~\ref{lem:k} and Lemma~\ref{CuadradosIndependientes}.

We fix a normalized generating set \genset of $\cC$ which will be modified throughout the proof to be adapted to one of the cases.
Let $\epsilon$ be as in Lemma~\ref{CuadradosIndependientes} and reorder the $z_i$'s such that those with equal square are consecutive and placed at the beginning of the list. Then $z_{2i-1}^2=z_{2i}^2=(z_{2i-1},z_i)$ for every $i=1,\dots,\epsilon$ (see the comment after Corollary~\ref{HadamardZs}).
Also, as in Lemma~\ref{CuadradosIndependientes} we set Let $V=\{y_1,\dots,y_{\delta},z_1,z_3,\dots,z_{2\epsilon-1},z_{2\epsilon+1},z_{2\epsilon+2},\dots,z_{\rho}\}$ and  $U=\{ u\in V : u^2 \not=\bu \}$.

\begin{enumerate}
\item
Assume $\epsilon=2$. Then $\delta=0$ and $\rho=4$ (Lemma~\ref{CuadradosIndependientes}). If either $z_1^2$ or $z_3^2$ equals $\bu$ we can assume that $z_1^2=\bu$. If $z_1^2$, $z_3^2$ and $\bu$ are pairwise different we can take $z'_1=z_1z_3$ and $z'_2=z_2z_4$ which is a new normalized generating set $z'_1,z'_2,z_3,z_4$ with ${z'_1}^2=\bu$. By Corollary~\ref{HadamardZs} $(z_i,z_j)=\langle z_j^2\rangle$, for $i=1,2$ and $j=3,4$. Hence, condition~\ref{rho=4} holds and $\bu\notin \gen{U}$, so $\sigma \geq \delta+\rho-\epsilon = 2$.

\item
Assume $\epsilon=1$ and $z_1^2=\bu$.
Then $\delta=0$ (Lemma~\ref{CuadradosIndependientes}).

If $\rho=2$ then condition~\ref{rho=arb-tau=0} holds. In this case $\bu\notin \gen{U}$, so $\sigma \geq \delta+\rho-\epsilon = \rho-1$.

Assume $\rho\ge 3$.
Then $U=\{z_3,\dots,z_{\rho}\}$.
By Corollary~\ref{HadamardZs}, for every $3\le i,j \le \rho$ we have $(z_i,z_j)=\be$ and therefore $\gen{(z_1,z_i),(z_2,z_i)}=\gen{z_i^2}$. By changing the generators $z_1$ and $z_2$ if necessary, we may assume that $(z_1,z_{\rho})=(z_2,z_{\rho})=z_{\rho}^2$. The new generating set is still  normalized.
We have two options: $\bu\not\in \gen{U}$ or $\bu \in \gen{U}$.
In the first case, $(z_1,z_i)=z_i^2$ for every $i\ge 3$ for otherwise $(z_1,z_iz_{\rho})=z_{\rho}^2\not\in \gen{(z_iz_{\rho})^2}$ in contradiction with Lemma~\ref{lemm:cuadradosHadamard}. Similarly $(z_2,z_i)=z_i^2$. Then condition~\ref{rho=arb-tau=0} holds.
We claim that in the second case, so when $\bu \in \gen{U}$, we have $(z_1,z_i)=\be$, $(z_2,z_i)=z_{i}^2$ for some $i\ge 3$. Otherwise
$(z_1,z_i)=(z_2,z_i)=z_i^2$ for every $i=3,\dots,\rho$. Then $(z_1z_2,z_i)=\be$ for every $j\ge 3$.
After reordering the $z_i$'s we may assume that $(z_3\cdots z_k)^2=\bu$. Then $(z_1z_2z_3\cdots z_k)^2=\be$ contradicting the fact that $z_1\dots z_k\not \in T(C)$. This proves the claim. So assume that $\bu\in \gen{U}$ and (after reordering the $z_i$) $(z_1,z_3)=\be$ and $(z_2,z_3)=z_3^2$. Then $(z_1,z_3z_{\rho})^2=z_{\rho}^2\not \in \gen{(z_3z_{\rho})^2}$ and therefore $(z_3z_{\rho})^2=\bu$.  If $4\le i <\rho$ then $(z_iz_{\rho})^2\ne (z_3z_{\rho})^2 = \bu$ and therefore $(z_1,z_i)=(z_2,z_i)=z_i^2$. Thus $(z_3z_i)^2=\bu=(z_3z_{\rho})^2$ which is not possible. This proves that $\rho=4$. Now we can construct a new generating set $\{z'_1=z_1z_2, z'_2=z'_2, z'_3=z_1z_3,z'_4=z_4\}$ and it is easy to check that $z'^2_1 = z'^2_2 = (z'_1,z'_2) = \bu \not = z'^2_3 = z'^2_4 = (z'_3,z'_4)$ and $(z'_i,z'_j) = z_3'^2$, for every $i=1,2$ and $j=3,4$.
Thus, $\epsilon=2$ which has been treated before.

\item
Assume $\epsilon=1$ and $z_1^2\not=\bu$.  We can have $\rho \geq 3$ or $\rho=2$.

In the first case, so if $\rho\ge 3$ then $(z_i,z_j)=(z_j,z_k)=\be$ for every $i\ne 2$ and every $j,k\ge 3$ such that $z_j^2=\be$,  by Corollary~\ref{HadamardZs}. As each $z_j$ is not central,  $z_j^2=\bu$ for some $j\ge 3$ and we may assume that $z_3^2=\bu$. After reordering coordinates one also may assume that $\gen{z_1,z_2}$ projects to $\cQ_8$ in the first $\frac{n}{8}$ coordinates. Then either $(z_1,z_3)\ne \be$ or $(z_2,z_3)\ne \be$ so, for instance, $(z_1,z_3)\ne\be$.  If $(z_2,z_3)= \be$ then replacing $z_2$ by $z_1z_2$, we can always assume that  $(z_2,z_3)\ne \be$. Then $(z_i,z_3)=z_i^2$ for $i=1,2$, by Lemma~\ref{lemm:cuadradosHadamard}. If it were $3<\rho$ then $z_4$ projects to an element of order at most 2 in the first $\frac{n}{8}$ coordinates and to an element of order 4 in the remaining coordinates and $(z_3,z_4)=z_4^2$.
Then $x_1,\dots,x_{\sigma};y_1,\dots,y_{\delta};z_1,z_2,z_3,z'_4=z_1z_4$ is a new generating set such that ${z'_4}^2 = z_1^2 z_4^2 = \bu = (z_3,{z'_4})$, with $\epsilon=2$, which is out of our initial assumption about $\epsilon=1$. Hence, we have $\rho=3$ with $z_3^2=\bu$ and $(z_i,z_3)=z_i^2$ for every $i=1,2$. Then $(z_1z_2z_3,z_1)=(z_1z_2z_3,z_2)=(z_1z_2z_3,z_3)=\be$. Therefore $z_1z_2z_3\in Z(\cC)$, a contradiction.

In the second case, so when $\rho=2$, assume $\delta >1$. We can reorder the coordinates in such a way that $\gen{z_1,z_2}$ projects to $\cQ_8$ in the first $\frac{n}{8}$ coordinates and then take two elements $y_1, y_2$ of order four which commutes with $z_1$ and $z_2$. The first $\frac{n}{8}$ coordinates of both $y_i$ must be of order at most two and so, $y_1y_2$ is of order two, contradicting the fact that $y_1,y_2$ are elements of a generating set. Hence, $\delta\leq 1$. In this case, note that if we take $\cA =\gen{x_1,\dots,x_{\sigma};y_1;z_1}$ then $A=\Phi(\cA)$ is a linear code. Indeed, the value of all swappers is $\be$ except for $[y_1,y_1]=y^2_1$ and $[z_1,z_1]=z^2_1$, which are also values belonging to $\cA$. Code $\cC =\gen{\cA,z_2}$ has only one possible swapper given by $[z_2,z_1]$ and so, $r(C)\leq m+2$. Then, condition~\ref{rho=2} holds.

\item
Finally assume $\epsilon=0$. Necessarily $z_i^2=\bu$ for some $i$, since $(z_i,z_j)=\be$ if $z_i^2,z_j^2$ and $\bu$ are pairwise different. We may assume that $z_1^2=\bu$. Then $(z_1,z_i)=z_i^2$ and $(z_i,z_j)=\be$ for every $i,j=2,\dots,\rho$.
We have $U=\{z_2,\dots,z_{\rho}\}$ and $\gen{U}$ is isomorphic to $\Z_2^{\rho-1}$. We have two options: $\bu\in \gen{U}$ or $\bu\not\in \gen{U}$.
In the first case, reordering $z_2,\dots,z_{\rho}$, we may assume that $z_2^2z_3^2\dots z_k^2=\bu$ for some $2<k \le \rho$. Then we change the set of generators by replacing $z_2$ by $z_2\dots z_k$.
Observe that we have passed from a normalized generating set with $\epsilon=0$ to a normalized one with $\epsilon=1$ and $(z_1,z_2)=z_1^2=z_2^2=\bu$. This case is already studied.
If $\bu\notin \gen{U}$ then $\delta=0$ and condition~\ref{rho=arb-tau=arb} holds.
Otherwise, there exists an element of order four, $y$ commuting with $z_1$ and $z_2$. But $(z_1,z_2)= z^2_2\not= \bu$, so the coordinates of order four in $z_2$ should coincide with the coordinates of order at most two in $y$ and then $y^2z^2_2=\bu$ or, the same $(yz_1)^2= z^2_2$.  Then $x_1,\dots,x_{\rho};y_1,\dots,y_{\delta};z'_1=yz_1,z_2,z_3,\dots,z_{\rho}$ is a new normalized generating set with $\epsilon=1$, and $(z'_1)^2=z_2^2=(z'_1,z_2)\not=\bu$, a case which has been treated before.
This finishes the proof.\end{enumerate}
\end{proof}

If $\cC$ is a subgroup of $\Z_2^{k_1}\times \Z_4^{k_2}\times \cQ^{k_3}$ such that $C=\Phi(\cC)$ is a Hadamard code and with a normalized set of generating vectors satisfying condition $i$ in Theorem~\ref{ClasificacionHadamard} we will say that $\cC$ is of \emph{shape} $i$.

\begin{coro}\label{Cotas}
Let $\cC$ be a subgroup of $\Z_2^{k_1}\times \Z_4^{k_2}\times \cQ^{k_3}$ of length $2^m$ and type $(\sigma,\delta,\rho)$ such that $C=\Phi(\cC)$ is a Hadamard code. Let $k=k(C)$ and $r=r(C)$. Then the following statements hold:
\begin{enumerate}
 \item $\left\lceil \frac{m}{2} \right\rceil \le \sigma \leq k \leq  m+1 \le r \le m+1+\binom{\delta+\rho}{2}$ and $\delta+\rho=m+1-\sigma \le \left\lfloor \frac{m+2}{2}\right \rfloor$, with one exception for a code with parameters $m=5,\sigma=2,\delta=0,\rho=4$.
% \item If $\cC$ is of shape~\ref{rho=2} then $r\le m+2+\binom{\left\lfloor \frac{m}{2}\right\rfloor}{2}$.
% \item If $\cC$ is of shape~\ref{rho=arb-tau=0} then $r\le m+1+\binom{\left\lfloor \frac{m}{2}\right\rfloor}{2}$.
% \item If $\cC$ is of shape~\ref{rho=arb-tau=arb} then $r\le m+1+\binom{\left\lfloor \frac{m+1}{2}\right\rfloor}{2}$.
% \item If $\cC$ is of shape~\ref{rho=4} then $r\le m+3$.
 \item  %$r(C) \leq m+2+\binom{\left\lfloor \frac{m}{2}\right\rfloor}{2}$.
$r\le \left\{\matriz{{ll} m+1+\binom{\frac{m+1}{2}}{2}, & \text{if } m \text{ is odd}; \\ m+2+\binom{\frac{m}{2}}{2}, & \text{if } m \text{ is even}.}\right.$

More precisely:
$$r-(m+1)\le \left\{\matriz{{ll}
\binom{\frac{m-1}{2}}{2}, & \text{if } m \text{ is odd and } \cC \text{ is of shape~\ref{rho=0}}; \\
1+\binom{\frac{m-1}{2}}{2}, & \text{if } m \text{ is odd and } \cC \text{ is of shape~\ref{rho=arb-tau=0} }; \\
\binom{\frac{m+1}{2}}{2}, & \text{if } m \text{ is odd and } \cC \text{ is of shape~\ref{rho=arb-tau=arb}}; \\
\binom{\frac{m}{2}}{2}, & \text{if } m \text{ is even and } \cC \text{ is of shape~\ref{rho=0} or 3}; \\
1+\binom{\frac{m}{2}}{2}, & \text{if } m \text{ is even and } \cC \text{ is of shape~\ref{rho=arb-tau=0} }; \\
1, & \text{if } \cC \text{ is of shape~\ref{rho=2}}; \\
3, & \text{if } \cC \text{ is of shape~\ref{rho=4}}; \\
}\right.$$
\end{enumerate}
\end{coro}

\begin{proof}
%Fix a generating set \genset of $\cC$ as in Theorem~\ref{ClasificacionHadamard} and let $\epsilon$ and $U$ be as in Lemma~\ref{CuadradosIndependientes}. Notice for all the shapes $\bu \not\in U$ and therefore $\sigma\ge \delta+\rho-\epsilon$.

Item 1. It is clear that $k \leq  m+1 =\sigma+\delta+\rho \le r$. From Lemma~\ref{lem:k} we have $\sigma \leq k$ and Lemma~\ref{lem:KernelSpan} gives $r \le m+1+\binom{\delta+\rho}{2}$.
Moreover  $\sigma \geq \delta+\rho-1$, except for the case $\epsilon=2$, which is  shape \ref{rho=4}. In this last case we also have $\delta=0$ and $\rho=4$, so $\sigma$ still fulfils the inequality $\sigma \geq \delta+\rho-1$, with the exception of $(m=5,\sigma=2,\delta=0,\rho=4)$.
Hence $m+1=\sigma+\delta+\rho\le 2\sigma+1$ (with the above exception) and therefore $\left\lceil \frac{m}{2} \right\rceil \le \sigma \leq k$.
Then $\delta+\rho=m+1-\sigma\le m+1-\frac{m}{2}=\frac{m+2}{2}$.

%Items 2, and 4 are improvements of the upper bound of $r$ which follow from item~5 of Lemma~\ref{CuadradosIndependientes}, Theorem~\ref{ClasificacionHadamard} and item 1 of this corollary. In the case of item 3 we know that $\sigma\geq \rho$, so $m+1=\sigma+\rho\geq 2\rho$ and $\rho \leq \frac{m+1}{2}$. Hence, $r(C) \leq \sigma+\rho+\binom{\rho}{2} \leq m+1 +\binom{\lfloor \frac{m+1}{2}\rfloor}{2}$ and the statement follows.

Item 2. Let $h=r-(m+1)$.

For shape~\ref{rho=0}(a), $m+1=\sigma+\delta\ge 2\delta$, hence $\delta-1\le \left\lfloor\frac{m-1}{2}\right\rfloor$. Hence, $h=r-(m+1) = (\sigma+\delta + \binom{\delta-1}{2} -(m+1)=\binom{\delta-1}{2}$. Thus, if $m$ is odd then $h\le \binom{\frac{m-1}{2}}{2}$ and if $m$ is even $h\le \binom{\frac{m-2}{2}}{2}<\binom{\frac{m}{2}}{2}$.

For shape~\ref{rho=0}(b), $m+1=\sigma+\delta\ge 2\delta+1$, so $\delta\le \left\lfloor\frac{m}{2}\right\rfloor$ and $h=\binom{\delta}{2}$. Thus, if $m$ is odd then $h \le \binom{\frac{m-1}{2}}{2}$ and if $m$ is even $h \le \binom{\frac{m}{2}}{2}$.

For shape~\ref{rho=arb-tau=0}, $m+1=\sigma+\rho\ge 2\rho-1$, so that $\rho-1\le \left\lfloor\frac{m}{2}\right\rfloor$, and $h\le 1+\binom{\rho-1}{2}$. Hence, if $m$ is odd then $h \le 1+\binom{\frac{m-1}{2}}{2}$ and if $m$ is even $h \le 1+\binom{\frac{m}{2}}{2}$.

For shape~\ref{rho=arb-tau=arb}, $\sigma\ge \rho$ and $m+1=\sigma+\rho\ge 2\rho$, so $\rho\le \left\lfloor\frac{m+1}{2}\right\rfloor$. Moreover, $h\le \binom{\rho}{2}$. Hence, if $m$ is odd then $h \le \binom{\frac{m+1}{2}}{2}$ and if $m$ is even $h \le \binom{\frac{m}{2}}{2}$.

For shape~\ref{rho=2}, $r\leq \sigma+\delta+\rho+1 \leq m+2$. Hence, $h=r-(m+1)\leq 1$.

Finally, the bound for shape~\ref{rho=4}, comes from Lemma~\ref{CuadradosIndependientes}.
\end{proof}

\begin{example}[A Hadamard $\Z_2\Z_4\cQ_8$-code with $\left\lceil \frac{m}{2} \right\rceil >  k$ and $m+1-\sigma>\left\lfloor\frac{m+2}{2}\right\rfloor$]\label{Hadamard3228}{\rm
By item 1 in Corollary~\ref{Cotas}, there are not Hadamard $\Z_2\Z_4\cQ_8$-codes of length $2^m$ such that neither $\left\lceil \frac{m}{2} \right\rceil < k$ nor $m+1-\sigma > \left\lfloor\frac{m+2}{2}\right\rfloor$, except perhaps for a code of parameters $(m=5,\sigma=2,\delta=0,\rho=4)$.
%if there is a Hadamard $\Z_2\Z_4\cQ_8$-code $C=\Phi(\cC)$ of length $2^m$ such that $\left\lceil \frac{m}{2} \right\rceil >  k(C)$, then $\cC$ satisfies condition~\ref{rho=4} of Theorem~\ref{ClasificacionHadamard}, its type is $(2,0,4)$ and $\sigma=k(C)=2$, and so $m=5$.
We present one example of such a code.

Consider the subgroup $\cC$ of $\cQ_8^8$ generated by
 \begin{eqnarray*}
 z_1 &=& (\ba,\ba,\ba,\ba,\ba,\ba,\ba,\ba), \\
 z_2 &=& (\bb,\bb,\ba\bb,\ba\bb,\bb,\bb,\ba\bb,\ba\bb), \\
 z_3 &=& (\ba,\ba,\ba^3,\ba^3,\one,\one,\ba^2,\ba^2), \\
 z_4 &=& (\bb,\bb^3,\ba\bb,\ba^3\bb,\one,\ba^2,\one,\ba^2).
 \end{eqnarray*}
Then $\cC$ is of type $(2,0,4)$, $z_1,z_2,z_3,z_4$ is a normalized generating system of $\cC$, which is of shape~\ref{rho=4}  and $\Phi(\cC)$ is a Hadamard code. Note that $k(C)=2=\frac{m-1}{2}<3=\left\lceil\frac{m}{2}\right\rceil$, $r(C)=8=m+3$ and $m+1-\sigma=4>3=\left\lfloor\frac{m+2}{2}\right\rfloor$.
}\end{example}

\begin{example}[The Hadamard codes of length 16]\label{Hadamar16}{\rm
Let $C$ be a Hadamard code of length 16 and let $r=r(C)$ and $k=k(C)$. As it was explained at the beginning of this section $(r,k)=\{(5,5),(6,3),(7,2),(8,1),(8,2)\}$. Of course if $(r,k)=(5,5)$ then $C$ is $\Z_2$-linear. If $(r,k)=(6,3)$ then $C$ is \zz and if $(r,k)\in \{(7,2),(8,1),(8,2)\}$ then $C$ is not a \zz linear code \cite{PRV06}.
In Proposition~\ref{Q8linearN16} we have exhibited a Hadamard $\cQ_8$-code of length 16 with $(r,k)=(7,2)$ and from item 2 of Corollary~\ref{Cotas} the upper bound for the rank of $\Z_2\Z_4\cQ_8$-codes of length $2^4$ is 7. Hence, the Hadamard codes of length 16 and rank 8 are not $\Z_2\Z_4\cQ_8$-codes.
}\end{example}

\section{Recursive constructions of Hadamard $\qq$-codes}\label{sec:Recursive}

In this section we present some methods to construct \qh codes from a given Hadamard code.

The complement of a binary vector $c$ is denoted $\overline{c}$. Observe that if $x\in \cG$ then $\overline{\Phi(c)} = \Phi(\bu c)$.

\subsection{From $\Z_2\Z_4$-linear Hadamard codes to Hadamard $\Z_4\cQ_8$-codes}

It is known \cite{PRV06} that for any $m$ we have $\lfloor \frac{m-1}{2}\rfloor$ nonequivalent $\Z_4$-linear Hadamard codes of binary length $n=2^m$. These codes can be characterized by the parameter $\delta$. Note that $\delta \in \{1,2,\ldots, \lfloor \frac{m+1}{2} \rfloor\}$, but the values $\delta =1,2$ give codes equivalent to the linear Hadamard. Also in~\cite{PRV06} there are described the $\Z_2\Z_4$-linear Hadamard codes (which are not $\Z_4$-codes). For any $m$ there are $\lfloor \frac{m}{2}\rfloor$ nonequivalent such codes of binary length $n=2^m$. As for the $\Z_4$-linear case, these codes can be characterized by the parameter $\delta \in \{0,1,2,\ldots, \lfloor \frac{m}{2} \rfloor\}$ and the values $\delta =0,1$ give codes equivalent to the linear Hadamard.

We begin by taking a $\Z_2\Z_4$-linear Hadamard code to obtain a Hadamard $\Z_4\cQ_8$-code.  Let $C=\Phi(\cC)$ be a $\Z_2\Z_4$-linear code, where $\cC$ is a subgroup of $\Z_2^{\alpha} \times \Z_4^{\beta}$.
Let $\xi_1: \Z_2 \longrightarrow \Z_4$ be the homomorphism defined by $\xi_1(i) = 2i$ and let $\xi_2:\Z_4 \longrightarrow \gen{\ba}\subseteq \cQ_8$ be the homomorphism defined by $\xi_2(i) = \ba^{i}$ and generalize those to a componentwise group homomorphism $\xi:\Z_2^{\alpha}\times \Z_4^{\beta}\rightarrow \Z_4^{\alpha}\times \cQ_8^{\beta}$. Let $\cC_q$ be the $\Z_4\cQ_8$-code $\xi(\cC)$ of binary length $2(\alpha+2\beta)=2n$. Code $\cC_q$ is a  $\Z_4\cQ_8$-code of the same type as $\cC$.

Assume that $\Phi(\cC)$ is a Hadamard code. Then the length of $\Phi(\cC_q)$ is $2n$ and all the codewords of $\Phi(\cC_q)$ have length $0$, $n$ or $2n$. However $\Phi(\cC_q)$ is not a Hadamard code since $|\cC_q|=|C|=2n$. Hence, to obtain a Hadamard code we need to double the cardinality of this code. We do that by taking $\cC^{(x)}=\gen{\cC_q,x}$ for an appropriate element $x\in \Z_4^{\alpha}\times \cQ_8^{\beta}$ of order 2 modulo $\cC_q$ which normalizes $\cC_q$ and $C^{(x)}=\Phi(\cC^{(x)})$. Then $\cC^{(x)}=\cC_q \cup\,x \cC_q$ and to make sure that $C^{(x)}$ is a Hadamard code we must choose $x$ so that
\begin{equation}\label{Condition}
\w(\Phi(x c))=n, \text{ for every }c\in \cC
\end{equation}

If $x$ has order 2 then, after reordering the coordinates we may assume that $x=(e_{l_1},\bu_{l_2})$. Then $C^{(x)}=\Phi(\cC_q)\cup \{(c_1,\overline{c_2}): (c_1,c_2)\in C\}$, where both $c_1$ and $c_2$ have length $n$. Observe that $\w(c_1,c_2)=\w(c_1)+n-\w(c_2)$.
Thus for $C^{(x)}$ to be a Hadamard code it is necessary that $\w(c_1)=\w(c_2)$ for every $(c_1,c_2)\in C$.

\begin{example}\label{Ejemplo16Lineal}{\rm
Take $\cC=\gen{(1,1,1,1),(2,0,1,3)}$, which is a $\Z_4$-linear code, but with the same codewords as the linear Hadamard code of length $8$. Then $\cC_q=\gen{(\ba,\ba,\ba,\ba),(\ba^2,\one,\ba,\ba^3)}$ and taking $x=(\one,\one,\ba^2,\ba^2)$ we obtain $C^{(x)}$, which is a Hadamard code of length 16 (with the same codewords as the binary linear Hadamard code of length 16).}
\end{example}

One way to ensure that $C^{(x)}$ is a Hadamard code is taking $x=(x_1,\dots,x_{\frac{n}{2}})$ with each $x_i\in \cQ_8\setminus \gen{\ba}$. Condition (\ref{Condition}) above is satisfied because for every $c\in \cC$, all the coordinates of $x c$ have order 4 and therefore $\w(\Phi(x c))=n$, as desired. The rank and dimension of the kernel of $C{(x)}$ depends on the election of $x$.

\begin{example}\label{Ejemplo16-72-63}{\rm
Take $\cC=\gen{(1,1,1,1),(2,0,1,3))}\subset \Z_4^4$, as in the Example~\ref{Ejemplo16Lineal}. If we choose $x=(\bb,\bb,\bb,\bb)$ then $C^{(x)}$ is again the (unique up to equivalence) binary linear code of length 16. If we take $x=(\bb,\ba\bb,\bb,\ba\bb)$ then $\cC^{(x)}$ is the group of Proposition~\ref{Q8linearN16}
and hence $C^{(x)}$ is the Hadamard $\cQ_8$-code of length 16 with rank $7$ and dimension of the kernel 2.
Finally, if we choose $y=(\bb,\bb,\bb,\ba^3\bb)$ then $C^{(y)}$ is a Hadamard $\cQ_8$-code of length 16, with rank 6 and dimension of kernel 3.
Hence the three Hadamard $\qq$-codes of length $16$ can be obtained applying our construction to $\cC$.
}\end{example}

The following theorem shows that most Hadamard $\Z_2\Z_4\cQ_8$-codes can be obtained with this construction.

\begin{theo}\label{theo:z2z4tohadamard}
Let $C'$ be a Hadamard $\Z_2\Z_4\cQ_8$-code. Assume that $C'$ is either of shape~\ref{rho=arb-tau=0} or \ref{rho=arb-tau=arb}. Then $C'$ is equivalent to $C^{(z)}$ for $C$ a \zz code and some $z$.
\end{theo}

\begin{proof}
 Assume that $C'=\phi(\cC')$ with $\cC'$ a subgroup of $\Z_2^{k_1}\times \Z_4^{\alpha}\times \cQ_8^{\beta}$ and let \genset a normalized generating set of $\cC'$ satisfying either condition~\ref{rho=arb-tau=0} or \ref{rho=arb-tau=arb} of Theorem~\ref{ClasificacionHadamard}. As $z_1^2=\bu$, we have $k_1=0$. Moreover, for shape~\ref{rho=arb-tau=0}, $(z_1,z_2)=\bu$ and therefore $\alpha=0$. Let
	$$\cC''=\left\{\matriz{{ll} \gen{x_1,\dots,x_{\sigma},z_1z_2,z_3,\dots,z_{\rho}}, & \text{for shape } \ref{rho=arb-tau=0}; \\
													 \gen{x_1,\dots,x_{\sigma};z_2,z_3,\dots,z_{\rho}}, & \text{for shape } \ref{rho=arb-tau=arb}.}\right.$$
Then $\cC''$ is an abelian subgroup of $\cC$ of index 2. Moreover, the projection on the $\Z_4$ part is contained in $\{0,2\}$. This is clear for shape~\ref{rho=arb-tau=0}. For shape \ref{rho=arb-tau=arb}, it is a consequence of $(z_1,z_i)=z_i^2$ for $i\ge 2$. After a suitable permutation on the $\cQ_8$-coordinates we may assume that $\cC\subseteq 2\Z_4^{\alpha}\times \gen{\ba}^{\beta}$ and therefore $\cC''=\xi(\cC)$ for a suitable subgroup $\cC$ of $\Z_2^{\alpha}\times \Z_4^{\beta}$ such that $C=\Phi(\cC)$ is a Hadamard code. Then $\cC'=\gen{\cC,z_1}$ and so $C'=C^{(z_1)}$.

Note that if $\cC'$ is of shape~\ref{rho=arb-tau=0}, then $\alpha=0$ and so it is equivalent to $C^{(z_1)}$ for $C$ a $\Z_4$-linear code.
\end{proof}

Notice that if $\cC$ is of shape~\ref{rho=4} then $\cC$ has not any abelian subgroup of index 2 and therefore $\cC$ can not be obtained with this type of construction.
%If $\cC$ is of \ref~{rho=arb-tau=0} then $\gen{x_1,\dots,x_{\sigma},y_1,\dots,y_{\delta},z_2}$ is an abelian subgroup of index 2. However it is not clear whether this can be obtained with this construction.
\bigskip

The \zz codes $\cC$ used in the last Theorem have length $n=2^{m}$, where $m+1=\sigma+\delta$ and $\sigma > \delta$ in the case we are dealing with \zz (non $\Z_4$-linear) codes (see~\cite{PRV06}). The parameters of the obtained code $\cC'$ after the construction in Theorem~\ref{theo:z2z4tohadamard} are $m'=m+1$, $\rho' = \delta+1$ and $\sigma' = \sigma$. Therefore, $m'=m+1=\sigma+\delta=\sigma'+\rho'-1$. Rank of $C'$ can be computed from rank of $C$ adding vector $z_1$ and all the swappers $[z_1,z_i]$, where $i\in \{1,\rho\}$, so $r(C')\leq r(C)+1+\rho'=\sigma+\delta+\binom{\delta}{2}+1+\rho' =\sigma'+\rho'+\binom{\rho'-1}{2}+\rho' =\sigma'+\rho'+\binom{\rho'}{2}$.

From Corollary~\ref{Cotas}, if $m$ is odd then the upper bound $r\le m+1+\binom{\frac{m+1}{2}}{2}$ can only be reached for Hadamard $\qq$-codes of shape~\ref{rho=arb-tau=arb}
or for shape~\ref{rho=4} , with $m=5$. For $m$ even the upper bound $r\le m+2+\binom{\frac{m}{2}}{2}$ only can be obtained with Hadamard $\qq$-codes of shape~\ref{rho=arb-tau=0}. For instance, for $m=4$ this maximum is 7 and it is reached by the code of Proposition~\ref{Q8linearN16} which is of shape~\ref{rho=arb-tau=0}. For $m=5$, the upper bound for the rank of a $\Z_2\Z_4\cQ_8$-code is 9. In the next example we will show that we can construct a $\Z_2\Z_4\cQ_8$-code with $m=5$ and rank 9, by using the latest construction.

\begin{example}\label{shape5}{\rm
Take the Hadamard $\Z_2\Z_4$-linear code $\cC$, with $m=4$ and parameter $\delta=2$. Code $\cC$ is generated by $(1,1,1,1\,|\,2,2,2,2,2,2)$, $(0,1,0,1\,|\,0,2,1,1,1,1)$, $(0,0,1,1\,|\,1,1,0,1,2,3)$ $\in \Z_2^4\times \Z_4^6$. Now, construct $\xi(\cC) \subset \Z_4^4\times \cQ_8^6$ generated by
\begin{align*}
x_1=&(2,2,2,2\,|\, \ba^2, \ba^2, \ba^2, \ba^2, \ba^2, \ba^2)\\
z_2=&( 0,2,0,2\,|\,  \one, \ba^2, \ba,\ba,\ba,\ba)\\
z_3=&(0,0,2,2\,|\, \ba,\ba, \one, \ba, \ba^2,  \ba^3)
\end{align*}

If we choose $z_1=(1,1,1,1,|\, \bb, \ba\bb, \bb, \ba\bb, \ba\bb, \ba^3\bb)$ then $C^{(z_1)}$ is a Hadamard code of length $32$, type $(3,0,3)$, shape~\ref{rho=arb-tau=arb}, rank $r=9$ and dimension of the kernel $k=3$.}
\end{example}

\subsection{The generalized Kronecker construction}

We give a generalization of the Kronecker construction of Hadamard matrices in the context of Hadamard $\Z_2\Z_4\cQ_8$-codes.
%This is the inspiration for a more general construction.

If $H$ is a Hadamard matrix then the Kronecker matrix of $H$ is $\mathcal{K}(H)=\left(\matriz{{cc} H & H \\ H & -H}\right)$, which is another Hadamard matrix. If $C$ is the Hadamard code associates to $H$ and $\mathcal{K}(C)$ is the Hadamard code associated to $\mathcal{K}(H)$ then $\mathcal{K}(C)$ is formed by the vectors of the form $(c,c)$ and $(c,\overline{c})$, with $c\in C$.

Let $\Delta:\cG\rightarrow \cG\times \cG$ be the diagonal map, i.e., $\Delta(x)=(x,x)$ for each $x\in \cG$.

Assume that $C=\Phi(\cC)$, for $\cC$ a subgroup of $\cG$. Then $\mathcal{K}(C)=\Phi(\mathcal{K}(\cC))$ where $\mathcal{K}(\cC)=\gen{\Delta(\cC),(1,\bu)}$. Moreover $k(\mathcal{K}(C))=k(C)+1$, $r(\mathcal{K}(C))=r(C)+1$ and if $\cC$ is of type $(\sigma,\delta,\rho)$ then $\mathcal{K}(\cC)$ is of type $(\sigma+1,\delta,\rho)$.

More generally, let $g$ be an element of $\cG$ of order 2 modulo $\cC$ which normalizes $\cC$, i.e., 	$$\cC^g=\cC \quad \text{ and } \quad g^2\in \cC.$$
Consider the subgroup $\mathcal{K}_g(\cC)=\gen{\Delta(\cC),(g,g\bu)}$ of $\cG \times \cG$. For example, $\mathcal{K}_g(\cC)=\mathcal{K}(\cC)$ if and only if $g\in \cC$. We claim that $\Phi(\mathcal{K}_g(\cC))$ is a Hadamard code.

First we have that $\Delta(\cC)$ is a subgroup of $\cG\times \cG$ of cardinality $2n$. Moreover, $(g,g\bu)^2=(g^2,g^2)\in \Delta(\cC)$; $\cC^g=\cC$ and $c\in \cC$ then $(g,g\bu)^{-1}(c,c)(g,g\bu)=(gcg^{-1},gcg^{-1})$. Therefore $\mathcal{K}_g(\cC)=\Delta(\cC)\cup \{(gc,g\bu c) : c\in \cC\}$ is a subgroup of $\cG\times \cG$ of cardinality $4n$. Furthermore, for every $c\in \cC$ we have $\w(\Phi(c,c))=2\w(\Phi(c))\in \{0,n,2n\}$ and $\w(\Phi(cg,c\bu g)) = \w(\Phi(cg),\overline{\Phi(cg)})=\w(\Phi(cg))+2n-\w(\Phi(cg))=2n$.

Moreover, $r(\mathcal{K}_g(C))=r(C)+1$ and $k(\mathcal{K}_g(C))\le k(C)+1$.
Assume that $\cC$ is of type $(\sigma,\delta,\rho)$. If $g$ has order $2$ then $\mathcal{K}_g(\cC)$ is of type $(\sigma+1,\delta,\rho)$.
If $g$ has order 4 and commutes with all the elements of $Z(\cC)$ then $\mathcal{K}_g(\cC)$ is of type $(\sigma,\delta+1,\rho)$. Finally, if $C_{Z(\cC)}(g)=\{x\in Z(\cC):xg=gx\}$ is of order $\sigma+\delta_1$ with $\delta_1<\delta$ then $\mathcal{K}_g(\cC)$ is of type $(\sigma,\delta_1,\rho+\delta-\delta_1+1)$.

\begin{example}{\rm
As an example of the above construction, from the code $C$ constructed in Proposition~\ref{Q8linearN16} and taking $g=(\bb, \ba \bb,\one,\one)\in \cQ_8^{4}$ we obtain a new code $\mathcal{K}_g(C)$ of binary length 32, which is a quaternionic Hadamard, non \zz code, with dimension of the kernel $2$ and rank $8$. It is equivalent to the code of Example~\ref{Hadamard3228}.}
\end{example}

\begin{example}\label{codi74}{\rm
Note that in some cases, when the size of the kernel of $C$ is strictly greater than the size of the center of $\cC$ it could happen that using the above generalized Kronecker construction the dimension of the kernel of the new code $\mathcal{K}_g(C)$ is lower than the original. As an example, take $\cC$ the subgroup of $\cQ_8^8$ generated by
\begin{align*}
z_1=&(\ba,\ba,\ba,\ba,\ba,\ba,\ba,\ba,)\\
z_2=&(\bb,\ba\bb,\bb,\ba\bb,\bb,\ba\bb,\bb,\ba\bb)\\
z_3=&(\ba^2,\one,\ba^2,\one,\ba^2,\one,\ba^2,\one)\\
z_4=&(\ba^2,\ba^2,\one,\one,\ba,\ba,\ba^3,\ba^3)
\end{align*}
The corresponding binary code $C=\Phi(\cC)$ is a Hadamard, non $\Z_2\Z_4$-linear code of length 32, type $(3,0,3)$, shape~\ref{rho=arb-tau=arb}, rank 7 and dimension of the kernel 4.

Take $g=(\ba^2,\ba^2,\one,\one,\bb,\ba\bb,\bb,\ba\bb)\in \cQ_8^8$ and construct  $\mathcal{K}_g(C)$, which turn out to be  a Hadamard, non $\Z_2\Z_4$-linear, code of length 64, type $(3,0,4)$, shape~\ref{rho=arb-tau=0}, rank 8 and dimension of the kernel 3.}
\end{example}

\subsection{Some final remarks}
Using the above constructions from an initial well known code (linear or \zz) we can construct several infinite families of $\Z_2\Z_4\cQ_8$-codes, which are not \zz.

We already mentioned that the Hadamard codes of length 16 can
be completely classified using the invariants given by the rank and the dimension of the kernel. However, in general, for larger lengths, we can find nonisomorphic Hadamard $\Z_2\Z_4\cQ_8$-codes with the same invariants.
\begin{example}\label{z2z4toq8}{\rm
As an example, consider code $C$ in Example~\ref{codi74}. It is a binary Hadamard, non $\Z_2\Z_4$-linear code of length 32, rank 7 and dimension of the kernel 4 . We also know a $\Z_2\Z_4$-linear code of length 32, rank 7 and dimension of the kernel 4 (item~\ref{rho=0} of Theorem~\ref{ClasificacionHadamard}). It is the code $C'$, where $C'=\Phi(\cC')$ and $\cC'$ is a subgroup of $\Z_2^8\times \Z_4^{12}$ generated by:
\begin{align*}
x_1=&(1, 1, 1, 1, 1, 1, 1, 1\,|\,2, 2, 2, 2, 2, 2, 2, 2, 2, 2, 2, 2)\\
x_2=&(0, 0, 0, 0, 1, 1, 1, 1\,|\, 0, 0, 0, 0, 0, 0, 2, 2, 2, 2, 2, 2)\\
y_1=&(0, 1, 0, 1, 0, 1, 0, 1\,|\, 0, 2, 1, 1, 1, 1, 0, 2, 1, 1, 1, 1)\\
y_2=&(0, 0, 1, 1, 0, 0, 1, 1\,|\, 1, 1, 0, 1, 2, 3, 1, 1, 0, 1, 2, 3)
\end{align*}}
\end{example}
\bigskip
 %It is worth to mention that for smaller values of $m$ ($m\leq 7$) and using {\sc MAGMA} software we have been able to construct Hadamard $\Z_2\Z_4\cQ_8$-codes with any allowable parameter for the dimension of the kernel and the rank (Corollary~\ref{Cotas}).

Note that the examples we wrote into the paper achieve almost all the shapes according Theorem~\ref{ClasificacionHadamard}. For instance, shape~\ref{rho=0} is satisfied for all well known $\Z_2\Z_4$-linear and $\Z_4$-linear Hadamard codes; code in Proposition~\ref{Q8linearN16} is of shape~\ref{rho=arb-tau=0}; code in Example~\ref{shape5} is of shape~\ref{rho=arb-tau=arb} and code in Example~\ref{Hadamard3228} is of shape~\ref{rho=4}. However, there is no any example of shape~\ref{rho=2}. We supply such an example below.

\begin{example}{\rm
Let $\cC$ be the $\Z_2^4\times \cQ_8$-code  $\langle (1,1,0,0|\ba) , (1,0,1,0|\bb) , (1,1,1,1|\ba^2) \rangle$. Code $\cC$ is of shape~\ref{rho=2} and, after the Gray map, code $C$ is a linear Hadamard code.

We can use a slightly variation of the Kronecker construction to obtain a shape~\ref{rho=2} non linear code with the maximum rank allowed for this shape.

First of all, we use the Kronecker construction to obtain the code $\cD=\mathcal{K}(\cC)$, which is generated by
\begin{align*}
x_1=&(1,1,1,1,1,1,1,1|\ba^2,\ba^2)\\
x_2=&(0,0,0,0,1,1, 1,1|\one,\ba^2)\\
z_1=&(1,1,0,0,1,1,0,0|\ba,\ba)\\
z_2=&(1,0,1,0,1,0,1,0|\bb,\bb)
\end{align*}
Code $D$ is a (linear) binary code of length 16, shape~\ref{rho=2} and dimension of the kernel and rank equal to the dimension of $D$, which is 5.

Finally, take the code ${\bar{\cD}}$ with generators $x_1, x_2, z_1$ and $${\bar{z}_2}=(1,0,1,0,1,0,1,0|\ba\bb,\bb).$$
It is straightforward to check that ${\bar{\cD}}$ is of shape~\ref{rho=2} and the binary code ${\bar D}=\Phi({\bar{\cD}})$ is of rank 6. Indeed, code ${\bar{\cD}}$ has a new swapper $[z_1,{\bar{z}_2}]=(0,0,0,0,0,0,0,0|\ba^2,\one)$ which did not exist in $\cD$.
}
\end{example}

\section*{Acknowledgment}
The authors wish to thank J. Borges and M. Villanueva for useful discussions and valuable comments.

\end{document}